\documentclass[sigconf,nonacm,balance=false]{acmart}

\usepackage[subtle]{savetrees}
\usepackage{nicefrac}
\usepackage{algorithm}
\usepackage{algorithmic}
\usepackage{subcaption}
\usepackage{multirow}
\usepackage{pgfplots}
\usepackage{pgfplotstable}
\usepackage{enumitem}
\usepackage{float}
\pgfplotsset{compat=1.18}
\usepgfplotslibrary{groupplots}
\usetikzlibrary{calc, matrix}

\newlength{\subplotwidth}
\newlength{\subplotheight}
\pgfplotsset{
  kvwzero/.style={blue, mark=*, solid, thick},
  kvwone/.style={green!60!black, mark=*, solid, thick},
  kvwthree/.style={red, mark=*, solid, thick},
  kvwfive/.style={purple, mark=*, solid, thick},
  prefixjasper/.style={orange, mark=square*, dashed, thick},
  prefixqdrant/.style={cyan!70!black, mark=square*, dashed, thick},
  appaxis/.style={
    width=0.31\textwidth, height=0.15\textheight,
    xlabel={Top-$k$}, xmin=2, xmax=103, xtick={5,10,20,50,100},
    xticklabel style={font=\tiny, rotate=35, anchor=east},
    tick label style={font=\tiny}, label style={font=\scriptsize},
    title style={font=\small}, grid=both, grid style={draw=gray!18},
    tick align=outside,
  },
}
\pgfplotsset{
  isttft/.style ={blue!70!black, mark=*,       solid,          thick, mark size=1.4pt},
  isq2ft/.style ={blue!70!black, mark=square*, densely dashed, thick, mark size=1.4pt},
  memttft/.style={red!75!black,  mark=*,       solid,          thick, mark size=1.4pt},
  memq2ft/.style={red!75!black,  mark=square*, densely dashed, thick, mark size=1.4pt},
  lataxis/.style={
    width=0.30\textwidth, height=3.9cm,
    xmode=log, log basis x=2, xmin=4.5, xmax=60,
    xtick={5,10,20,50}, xticklabels={5,10,20,50}, xlabel={Top-$k$},
    ymode=log, ymin=13, ymax=320,
    ytick={20,40,80,160,320}, yticklabels={20,40,80,160,320},
    tick label style={font=\footnotesize}, label style={font=\footnotesize},
    title style={font=\small}, grid=both, grid style={draw=gray!18},
  },
}
\pgfplotsset{
  memtput/.style={red!75!black,   mark=*,        solid,          thick, mark size=1.4pt},
  gputput/.style={blue!70!black,  mark=square*,  solid,          thick, mark size=1.4pt},
  cputput/.style={teal!65!black,  mark=triangle*,densely dashed, thick, mark size=1.6pt},
  tputaxis/.style={
    width=0.30\textwidth, height=3.9cm,
    xlabel={Concurrent users}, xmin=5, xmax=105, xtick={10,25,50,75,100},
    ymin=0, ymax=145, ytick={0,40,80,120},
    tick label style={font=\footnotesize}, label style={font=\footnotesize},
    title style={font=\small}, grid=both, grid style={draw=gray!18},
  },
}
\usepackage{mathtools}
\usepackage{xspace}
\usepackage[noabbrev]{cleveref} % load after hyperref/amsthm
\usepackage{balance}
\crefname{appendix}{Appendix}{Appendices}
\Crefname{appendix}{Appendix}{Appendices}

\theoremstyle{acmplain}
\newtheorem{theorem}{Theorem}
\newtheorem{proposition}[theorem]{Proposition}
\newtheorem{corollary}[theorem]{Corollary}
\theoremstyle{acmdefinition}
\newtheorem{definition}[theorem]{Definition}

\definecolor{picolor}{HTML}{4ECDC4}
\definecolor{kvcolor}{HTML}{FF6B6B}
\definecolor{nmcolor}{HTML}{888888}

\newcommand{\sysname}{\textsc{InferScale}\xspace}
\newcommand{\locomo}{LoCoMo\xspace}

\newcommand{\para}[1]{\smallskip\noindent\textbf{#1.}}

\setcopyright{rightsretained}
\begin{document}

\title{\sysname: GPU-Native KV Injection for Personalized LLM Serving}

\author{Peter Li}
\affiliation{%
  \institution{Northeastern University}
  \city{Boston}
  \state{Massachusetts}
  \country{USA}
}
\email{li.pet@northeastern.edu}

\author{Prashant Pandey}
\affiliation{%
  \institution{Northeastern University}
  \city{Boston}
  \state{Massachusetts}
  \country{USA}
}
\email{p.pandey@northeastern.edu}

% Abstract must precede \maketitle in acmart.
\begin{abstract}

Large language models are increasingly deployed with persistent personalized
context, such as accumulated memory profiles or long conversation histories,
that is shared across a user's many requests. Production memory systems (e.g.,
Mem0, MemGPT, and Zep) retrieve a relevant subset of this memory and inject it
into the prompt, forcing the serving engine to repeatedly prefill the same
content. As the retrieval budget grows, time-to-first-token (TTFT) increases
even though the underlying memory is reused across requests.

We present \sysname, a GPU-native LLM memory system that replaces repeated
prompt prefilling with reusable KV state. \sysname precomputes each memory
fact's KV representation, stores it alongside a semantic embedding on the GPU,
retrieves relevant facts at serving time, and injects their KV directly into
vLLM's paged cache.
% We show that, under causal self-attention, the KV
% representation of static memory is independent of the query and can therefore
% be reused exactly. 
To support dynamically assembled memories under rotary
position embeddings, we introduce \emph{Chunked RoPE}, which stores keys before
rotation and applies their serving-time positions during injection. 
However, encoding memory facts independently omits the cross-fact context available
during joint prefilling. We mitigate this with \emph{Context-Window Encoding},
which encodes each memory fact together with a small window of preceding
conversation context while caching only the target fact's KV.

\sysname is implemented through vLLM's KV-connector interface, requiring
neither engine modifications nor model fine-tuning.
Across three open-weight models on \locomo, \sysname keeps TTFT nearly constant
as the retrieval budget increases. On Llama-3.1-8B, TTFT increases by
only 4\% from $k=5$ to $k=50$ (16.6--17.3\,ms), compared with 106\% for Mem0, a
state-of-the-art memory system (33.2--68.3\,ms). At $k=50$, \sysname reduces
TTFT by 72--79\% (3.6--4.8$\times$), achieves 60.3\% accuracy versus 63.3\% for
Mem0 without serving-time recomputation, and delivers 3.7--4.5$\times$ the
throughput under concurrent load. Scaling the retrieval budget to $k=400$
widens the TTFT reduction to 13.9$\times$, and KV injection, rather than
GPU-native retrieval, becomes the dominant source of the end-to-end latency
reduction (62\%). These results demonstrate that reusable KV
state decouples memory-conditioned serving latency from retrieved-context size
while preserving application quality.

\end{abstract}


% \keywords{LLM serving, KV cache, paged attention, personalized memory,
% retrieval-augmented generation, inference systems}

\maketitle

\section{Introduction}
\label{sec:intro}

Production LLM applications increasingly maintain \emph{persistent
user-specific context}, such as accumulated memories or long conversation
histories, that is reused across many requests.
% Production deployments of large language models increasingly condition each
% request on a substantial \emph{static context} shared across a user's many
% queries: an accumulated memory profile, a retrieved document set, a cached
% system prompt, a library of few-shot exemplars. 
Because prepending a user's entire history to every request incurs prefill cost
quadratic in its length, production memory systems such as
Mem0~\cite{chhikara2025mem0}, Zep~\cite{zep2025}, and MemGPT/Letta~\cite{letta}
manage this context by retrieving a relevant top-$k$ subset of memory
\emph{facts}, short natural-language statements distilled from the user's
history,
% the atomic unit these systems store and retrieve, 
per query and \emph{injecting them as prompt text}.
Retrieval avoids prefilling an entire memory store, but it does not eliminate
the cost of conditioning: every retrieved fact is still re-prefilled on
every request. Consequently, time-to-first-token (TTFT) grows with the amount
of memory retrieved, even when the retrieved memories are identical across many
requests.

% Retrieval keeps the prompt short, but it does not make injection free. The
% retrieved memory is prepended to the prompt and \emph{re-prefilled on every
% request}, so per-request latency grows with the amount of context retrieved:
% the more memory a query needs, the slower it is served. This serving cost is
% what we target---and, unlike the model's accuracy, it is avoidable.

% \todo{add citations for prefix caching}
Existing KV-reuse techniques such as prefix
caching~\cite{kwon2023pagedattention} reduce repeated prefilling for identical
prompt prefixes, but they assume reused tokens occupy fixed positions in the
prompt. Persistent-memory systems violate this assumption: different subsets of
memory are assembled dynamically for each request, and their positions depend
on the retrieval result. Therefore, today's memory systems continue to pay
the full prefill cost of retrieved context despite extensive reuse across
requests.

\para{Inject at the attention layer, not the token layer}
Our key observation is simple: under causal attention, the key--value (KV)
representation of a static memory fact depends only on the fact itself and is
independent of the query. Therefore, a fact can be encoded once and
reused across arbitrarily many future requests instead of being repeatedly
re-prefilled. Concretely, this collapses the per-request attention cost of
conditioning from $O((m{+}q)^2)$ to $O(q(m{+}q))$ for $m$ memory and $q$ query
tokens, eliminating the $O(m^2)$ term a memory system otherwise re-pays on every
one of a user's requests.

This observation suggests a different serving primitive. Rather than injecting
retrieved memory as prompt tokens (\emph{prompt injection}), we inject its KV
representation directly into the attention cache (\emph{KV injection}). We show
that, for a fixed context injected at its intended positions, this is not an
approximation but an exact equivalence (\Cref{thm:equivalence}): KV injection
produces the same hidden states and output distribution at all query positions
as prompt injection, for any positional encoding, attention variant, and
feed-forward architecture.

However, achieving KV injections comes with two challenges.
The \emph{first} challenge is positional encoding. Retrieved facts must often
appear at positions different from where they were originally encoded, making
naïve KV reuse incorrect under rotary positional embeddings (RoPE). We address
this challenge with \emph{chunked RoPE}: storing keys before rotary position
encoding and applying the appropriate rotation at insertion time allows a fact
encoded once to be injected at any prompt position while preserving exactly the
same attention behavior as prompt injection (\Cref{thm:chunked-rope}).

The \emph{second} challenge is preserving the accuracy of prompt injection.
Encoding each fact independently eliminates repeated prefilling, but it also
removes the interactions between neighboring facts that would naturally arise
if they were jointly encoded as part of a prompt. We address this with
\emph{context-window encoding}, which encodes each fact together with a small
window of preceding conversation turns while caching only the target fact's KV
representation. This recovers nearly the accuracy of prompt injection without
requiring serving-time recomputation, cache re-encoding, or model fine-tuning.
In short, the injection mechanism itself, i.e., KV injection with chunked RoPE, is
\emph{exact}; the only approximation \sysname introduces is encoding facts
independently, which context-window encoding largely recovers.

\para{\sysname}
We realize these ideas in \sysname,\footnote{Source code at
\url{https://github.com/saltsystemslab/InferScale}.} a GPU-native
retrieve-and-inject serving engine for persistent memory. \sysname is a
data-systems co-design: it maintains two GPU-resident indices over the same
memory, keyed by a shared identifier, an approximate-nearest-neighbor index
over semantic embeddings for retrieval and a pre-RoPE KV store for injection, and
manages them across an GDDR/host-DRAM hierarchy. Keeping both on-device means
retrieved memory never crosses the PCIe bus and only the $q$ query tokens are
prefilled. At serving time,
it embeds the query, retrieves the relevant facts, applies chunked RoPE
to assign their serving-time positions, and injects the resulting KV directly
into vLLM's~\cite{kwon2023pagedattention} paged KV cache through the
KV-connector interface. \sysname requires no model retraining and no modifications
to the serving engine.
% , and no changes to existing retrieval pipelines.
\Cref{fig:cpu-gpu} shows the high-level design of \sysname and how it differs
from existing LLM memory systems such as Mem0.  

\begin{figure}[t]
  \centering
  \includegraphics[width=\columnwidth]{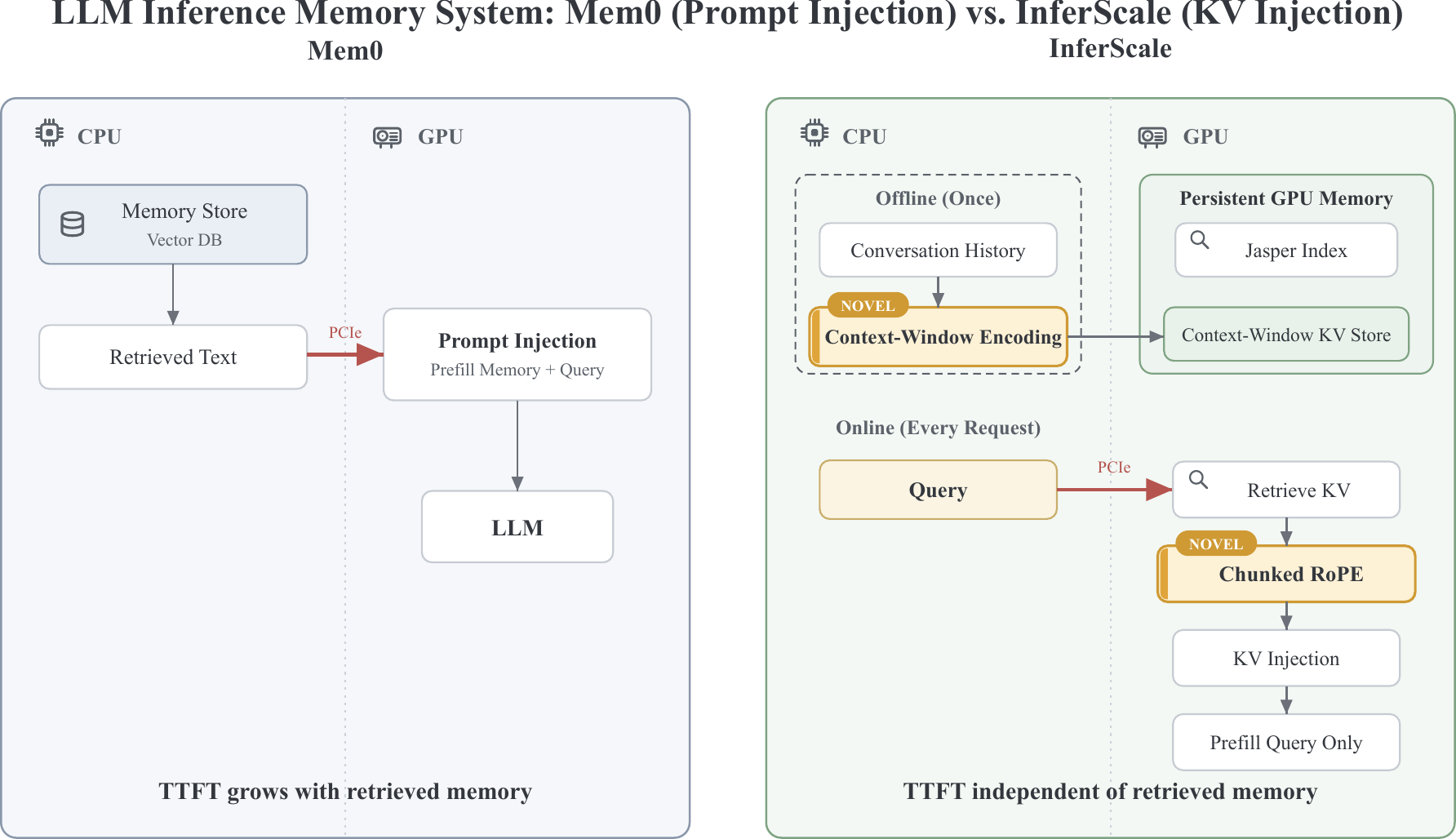}
  \caption{Where memory lives, and what crosses the PCIe bus. \textbf{(a)}~In
  Mem0 (prompt injection) the vector database and memory store sit on the CPU. Each
  request ships the retrieved memory to the GPU as prompt \emph{text}, and the
  GPU re-prefills all $m$ retrieved tokens, so TTFT grows with
  the amount retrieved. \textbf{(b)}~\sysname keeps the vector index (Jasper) and the
  pre-RoPE KV store on the GPU. Only the $q$ query tokens cross to the GPU and
  only they are prefilled, while retrieved memory is injected as KV and never
  leaves the device.}
  \label{fig:cpu-gpu}
\end{figure}

% We realize this as \sysname, a GPU-native
% retrieve-and-inject serving engine. 
% Each memory chunk is encoded twice and
% linked by a shared identifier: a semantic embedding held in a GPU-resident
% vector index (Jasper) for retrieval, and a pre-RoPE KV tensor held on HBM for
% injection. At serving time \sysname embeds the query, retrieves the top-$k$
% chunks, composes their KV with chunked RoPE, and scatter-copies the result into
% the paged cache, all on-GPU, and implemented entirely through
% vLLM's~\cite{kwon2023pagedattention} KV-connector interface without modifying
% the engine.

\para{Contributions}
% \label{sec:intro:contributions}
Our contributions are:

\begin{itemize}[leftmargin=1.2em]

\item \textbf{Retrieve-and-inject inference.}
We show that retrieved memory need not be re-prefilled on every request.
Instead, its key--value (KV) representation can be injected directly into the
attention cache. We prove that KV injection is equivalent to prompt
injection at query positions for any positional encoding scheme, attention
variant, and feed-forward architecture (\Cref{thm:equivalence}).

\item \textbf{Chunked RoPE for position-independent KV reuse.}
Retrieved memory must often be placed at prompt positions different from where
it was encoded. We introduce \emph{Chunked RoPE}, which stores keys before
rotary positional encoding and re-rotates them at insertion time, allowing a
fact encoded once to be injected at arbitrary prompt positions while preserving
the attention produced by prompt injection (\Cref{thm:chunked-rope}).

\item \textbf{Context-window encoding for accurate KV reuse.}
Encoding retrieved facts independently sacrifices the attention between
neighboring facts and reduces model accuracy. We introduce \emph{Context-Window
Encoding}, which encodes each fact together with a small window of preceding
conversation turns while caching only the target fact's KV representation. This recovers nearly the
accuracy of prompt injection without serving-time recomputation, cache
re-encoding, or model fine-tuning.

\item \textbf{InferScale: a GPU-native retrieve-and-inject serving engine.}
We design and implement \sysname{} as a vLLM KV-connector plugin that combines
GPU-resident semantic retrieval with reusable KV storage, requiring no
modifications to the serving engine, no model fine-tuning, and no changes to
existing retrieval pipelines (\Cref{sec:design}).

% \item \textbf{Efficient serving with preserved quality.}
% Across three open-weight models on the LoCoMo benchmark, \sysname{} reduces
% time-to-first-token by $3.6$--$4.8\times$ and increases serving throughput by
% $3.7$--$4.5\times$ over Mem0, while context-window encoding recovers
% nearly its accuracy (\Cref{sec:eval}).

\end{itemize}

\para{Results}
% \label{sec:intro:results}
By replacing repeated prompt prefilling with KV injection, \sysname makes
serving latency nearly invariant to the amount of retrieved memory. On the
\locomo benchmark~\cite{locomo}, as retrieval grows from $k{=}5$ to $k{=}50$ the
vLLM engine TTFT for Llama-3.1-8B rises by only 4\% (16.6 to
17.3\,ms), whereas Mem0 rises by 106\% (33.2 to 68.3\,ms). This
trend is consistent across three open-weight models, giving 72--79\% lower TTFT
at $k{=}50$, a $3.6$--$4.8\times$ speedup. The gap widens as the retrieval
budget grows further: at $k{=}400$, Mem0's TTFT reaches $343$\,ms while
\sysname's stays at $25$\,ms ($13.9\times$). GPU-native retrieval accounts for
most of the end-to-end saving at $k{\le}50$. However, KV injection's share grows
with $k$, from $30\%$ at $k{=}100$ to $62\%$ at $k{=}400$, making it the
dominant source of \sysname's speedup at large retrieval budgets
(\Cref{sec:eval:scaling-k}).

Encoding facts independently introduces an accuracy trade-off, as
interactions between neighboring retrieved facts are not captured during
encoding. Context-window encoding largely eliminates this trade-off by encoding
each fact together with a small window of preceding conversation turns while
caching only the target fact. It eliminates the accuracy collapse of
independent encoding, making accuracy flat-to-rising as more memory is
retrieved, achieving 60.3\% on \locomo at $k{=}50$ compared to Mem0's 63.3\%,
while matching or exceeding Mem0 at smaller retrieval budgets.
Under concurrent load, \sysname's throughput exhibits near-linear scaling with
the number of users, reaching $3.7$--$4.5\times$ that of Mem0 at $100$ users.

\sysname's per-conversation KV store ($1.8$--$4.8$\,GB) can also be offloaded
from GPU memory to host DRAM and streamed over PCIe on demand, lifting the HBM
capacity bound at negligible cost: at $k{=}50$ on Llama-3.1-8B, offloading adds
only ${\sim}3$\,ms to engine TTFT and keeps end-to-end query-to-first-token
(${\approx}100$\,ms) still $2.3\times$ faster than Mem0 ($236$\,ms), while
leaving throughput essentially unchanged.

These results distinguish \sysname from three adjacent lines of work.
Production memory systems (Mem0~\cite{chhikara2025mem0},
MemGPT/Letta~\cite{letta}, Zep~\cite{zep2025}) retrieve relevant context but
still inject it as prompt text, so latency grows with retrieved-context size.
Prefix- and block-level KV-reuse techniques
(Block-Attention~\cite{blockattention}, CacheBlend~\cite{yao2024cacheblend},
LazyAttention~\cite{lazyattention}, LMCache~\cite{lmcache}) avoid re-prefilling
reusable context, including at non-prefix positions, but do not target
persistent, semantically retrieved per-user memory. To restore cross-chunk
context, they typically require model fine-tuning or recompute selected tokens
\emph{online}, on every request. Sparse
and retrieval attention (RetroInfer~\cite{retroinfer2025},
RetrievalAttention~\cite{retrievalattention2024}) instead reduce the
\emph{decode-time} cost of attending over a single long context. In contrast,
\sysname applies KV reuse to persistent, semantically retrieved memory. It
couples GPU-native semantic retrieval with relocatable KV injection to
eliminate the repeated \emph{prefill} of retrieved memory. Context-window
encoding restores cross-fact context \emph{offline}, once per fact,
% provably exactly, 
with no fine-tuning, engine changes, or attention recomputation, while
preserving nearly the accuracy of prompt injection.

\section{Background and Motivation}
\label{sec:background}
In this section, we introduce the components of modern LLM serving that
\sysname builds upon. We first review the prefill/decode execution model and KV
cache management, which explain why repeatedly injecting retrieved memory
incurs high serving latency. We then describe retrieval-based memory systems,
whose token-layer interface motivates \sysname's attention-layer injection.
Finally, we review rotary position embeddings and GPU-native vector indexing,
which enable our two key optimizations---Chunked RoPE for position-independent
KV reuse and Context-Window Encoding for accurate memory representations.

% ============================================================

\subsection{LLM inference: prefill and decode}
\label{sec:bg:pipeline}
An autoregressive LLM generates a response one token at a time, each token
attends to the full sequence of tokens before it. Serving a request
proceeds in two distinct phases. In the \emph{prefill} phase, the
model consumes the entire input prompt in a single parallel forward pass: for
each of the $n$ prompt tokens and each layer it computes query, key, and value
(KV) projections, runs self-attention, and produces the first output token.
Because every token attends to all preceding ones, prefill performs
$\Theta(n^2)$ attention work and is compute-bound. Prefill is a major
model-execution component of \emph{time-to-first-token} (TTFT). In the \emph{decode} phase,
the model emits the remaining tokens one at a time. Each new token attends to
all previous tokens and appends its own KV, so per-step work grows with the
running context length and the phase is bound by memory bandwidth rather than
compute.

To avoid recomputing the KV of earlier tokens at every decode step, serving
engines store them in a \emph{KV cache}: prefill populates the cache for the
prompt, and decode reads it and appends one entry per generated token. Caching
makes each decode step cheap, but the prefill phase bears the cost of input
context that scales quadratically when no reusable prefix-cache entry is available.
% it also means the cost of the prompt is paid up front, in full, during
% prefill. The problem central to this paper is that \emph{any context
% prepended to the prompt, such as user memory, is re-prefilled on every
% request}, at a cost that grows with its length, even when that context is
% identical across requests.

\subsection{Paged attention and KV cache management}
Because the KV cache grows with sequence length and differs across concurrent
requests, its memory management is a central concern in LLM serving. Modern
engines such as vLLM~\cite{kwon2023pagedattention} use \emph{paged attention}:
each request's KV cache is divided into fixed-size \emph{blocks} (e.g., 16
tokens) allocated on demand from a shared pool. This avoids the fragmentation
of contiguous allocation and enables \emph{prefix caching}, i.e., reusing
blocks across requests that share a common prompt prefix, and continuous
batching. Prefix caching reuses contiguous prompt prefixes but cannot reuse
dynamically retrieved memory assembled at different positions. \sysname builds
on this substrate: rather than let prefill produce the memory tokens' KV, it
inserts precomputed KV directly into a request's paged blocks
(\Cref{sec:design}), so those blocks are populated from the cache instead of
by computation.

\subsection{Retrieval-based memory and its serving cost} 
\label{sec:bg:memory}
Production memory systems share one architecture.
\textbf{Mem0}~\cite{chhikara2025mem0} extracts facts from conversations, stores
them in a vector database, and retrieves the top-$k$ relevant memories by
embedding similarity. \textbf{Zep}~\cite{zep2025} maintains a temporal
knowledge graph and retrieves relevant subgraphs.
\textbf{Letta/MemGPT}~\cite{letta} pages facts between a main context and
archival storage.
% ; \textbf{ChatGPT} and \textbf{Claude} combine extracted facts, summaries,
% and recent history. 
In every case the retrieved memory is serialized into \emph{prompt text},
concatenated with the query, and processed through full prefill. All these
systems operate at the \emph{token layer}.

Retrieval is both effective and necessary: because prefill cost is quadratic,
prepending a user's entire history to every request is infeasible, so keeping
the prompt short is the point. But token-layer injection carries two serving
costs. First, retrieved memories are re-prefilled on every request, causing
latency to grow with the amount of retrieved context. Second, retrieved text
must be transferred from CPU memory to the GPU before prefilling. \sysname
eliminates both costs by storing the retrieval index and reusable KV
representations on the GPU and injecting memory directly into the KV cache.

% But token-layer injection carries two costs
% (\Cref{fig:cpu-gpu}a). The retrieved memory is \emph{re-prefilled on every
% request}, so per-request latency grows with the number of retrieved tokens;
% and, to fit a prompt budget, these systems further compress raw context into
% lossy extracted facts. \sysname keeps retrieval but moves it and the memory
% store onto the GPU and injects memory at the \emph{KV cache layer}
% (\Cref{fig:cpu-gpu}b): retrieved memory never crosses the PCIe bus and only the
% query is prefilled, which removes the first cost, while injecting raw retrieved
% turns rather than extracted facts changes the character of the second.

\subsection{Rotary position embeddings}
\label{sec:bg:rope}
Most modern decoders encode position with rotary position embeddings
(RoPE)~\cite{su2021roformer}, which rotate the query and key vectors at position
$p$ by a fixed block-diagonal rotation $R_p$. Because $R_a^\top R_b = R_{b-a}$,
the attention score $\langle R_a q, R_b k\rangle = \langle q, R_{b-a} k\rangle$
depends only on the \emph{relative} position $b-a$. Two things to note here:
% consequences matter here: 
the rotation is applied \emph{after} the key projection $W_K h$, so a key
can be stored before rotation and rotated later; and the score a key receives
depends only on where it sits relative to the query, not on its absolute index.
\sysname's chunked RoPE (\Cref{sec:chunked-rope}) rests on both.

\subsection{Contextual encoding of memory}
Retrieved memories are commonly stored as independently retrievable units. Mem0,
for example, extracts discrete \emph{facts} from conversation turns. While this
organization supports efficient semantic retrieval, transformer representations
are contextual, i.e., a token's hidden state depends on preceding tokens
through self-attention. Encoding a fact independently removes conversation
context that can disambiguate its meaning, reducing downstream answer quality
after KV injection.
This tension between retrieval granularity and contextual encoding motivates our
\emph{context-window encoding} technique, which preserves local context while
maintaining independently reusable fact representations.

\subsection{Jasper: A GPU-native vector index}
\label{sec:bg:jasper}
Retrieval over embeddings is an approximate nearest-neighbor (ANN) search
problem, conventionally served from a CPU-resident vector index such as
HNSW~\cite{malkov2020hnsw} or Faiss~\cite{johnson2019faiss}. \sysname instead
uses \emph{Jasper}~\cite{mccoy2026gpu}, a GPU-resident graph-based ANN index (in
the spirit of GPU indices such as CAGRA~\cite{ootomo2023cagra}):
it builds a proximity graph over the fact embedding vectors in GPU memory and
answers top-$k$ queries with a beam search on the device. Keeping both the
vectors and the search on the GPU lets retrieval run on the same device as the
pre-RoPE KV store and the serving engine, so a request never leaves the GPU
between embedding its query and injecting the retrieved KV. We access Jasper as
the vector backend of Mem0; as \Cref{sec:eval} shows, because Mem0 runs Qdrant
in exact-search mode the two indices return near-identical results and answer
quality changes only marginally, so the choice of index affects serving latency
rather than accuracy.
Jasper's device memory footprint is negligible compared to the space required
by model weights and KV cache. We present and discuss the memory footprint of
Jasper index in \Cref{sec:eval}.  

\section{\sysname Design}
\label{sec:design}
% ============================================================

\begin{figure*}[t]
    \centering
    \includegraphics[width=\textwidth]{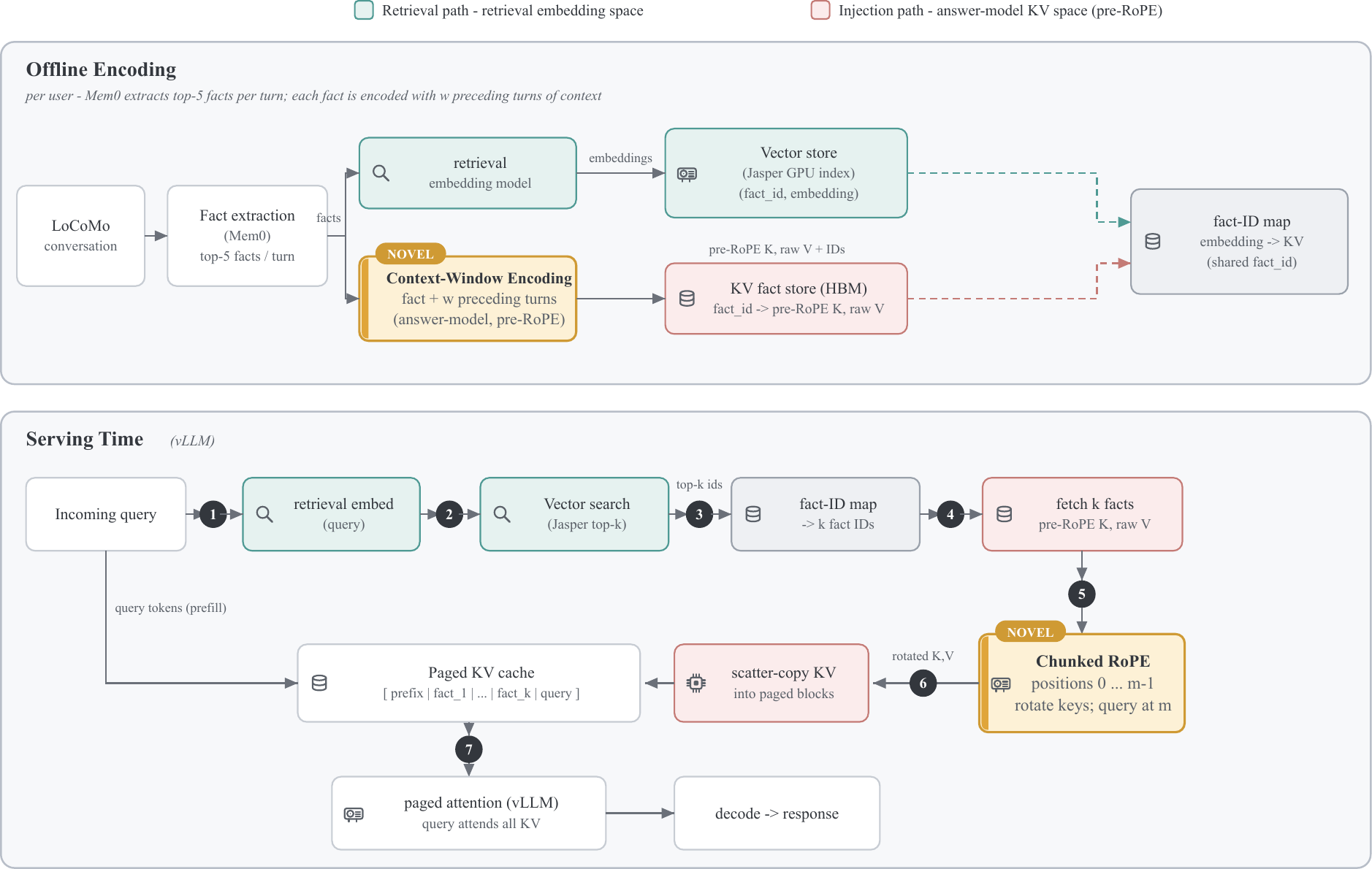}
    \caption{\sysname architecture; the preprocessing pipeline follows
    Mem0~\cite{chhikara2025mem0}. \textbf{Offline}, each extracted fact is
    encoded twice under a shared \texttt{fact\_id}, a retrieval embedding in the
    GPU vector store (retrieval space) and a pre-RoPE key/value tensor from the
    answer model (injection space), with a window of $w$ preceding
    turns. \textbf{At serving time}, the top-$k$ retrieved facts are resolved to
    their KV, positioned by chunked RoPE, and scatter-copied into the paged cache
    ahead of the query. \Cref{sec:design} details each stage.}
    \label{fig:system_arch}
\end{figure*}

\sysname is a GPU-native retrieve-and-inject serving engine implemented as a
vLLM \emph{KV connector}~\cite{kwon2023pagedattention}, allowing it to integrate with stock
vLLM without modifying the serving engine. The design separates memory
processing into two phases. An \emph{offline phase}, executed once when memory
is created, constructs reusable KV representations that are stored on the GPU.
An \emph{online serving phase}, executed for every request, retrieves the
relevant facts, adapts them to their serving-time positions, and injects
them directly into the attention cache before the query is prefilled. This
separation eliminates repeated prefilling while preserving compatibility with
existing retrieval pipelines and serving infrastructure.
\Cref{fig:system_arch} gives an overview of both phases.

The offline phase consists of Context-Window Encoding and construction of a
persistent GPU memory store. The online phase performs semantic retrieval,
Chunked RoPE composition, and KV injection through the vLLM KV-connector
interface. We describe each stage below and defer the correctness proofs to
Section~\ref{sec:theory}.

% \sysname separates memory processing into an offline phase, performed once
% when memory is created, and an online serving phase, executed for every
% request. The offline phase constructs reusable memory representations, while
% the online phase retrieves the relevant memories and injects them directly into
% the attention cache. This separation eliminates repeated prefilling while
% requiring no changes to vLLM or the model itself.
%
% \sysname is a GPU-native retrieve-and-inject serving engine built as a vLLM
% \emph{KV connector}~\cite{kwon2023pagedattention}---the interface vLLM exposes for
% externally managed cache blocks---so it requires no changes to the engine
% itself. It has three parts: (1)~an offline \emph{per-turn chunk encoder} that
% stores keys in pre-RoPE form; (2)~a dual memory store keyed by
% \texttt{turn\_id}---a GPU-resident vector index (Jasper,
% \S\ref{sec:bg:jasper}) for retrieval and a pre-RoPE KV store on HBM for
% injection; and (3)~a serving-time pipeline that retrieves the top-$k$ chunks,
% composes their KV with chunked RoPE, and scatter-copies the result into the
% paged cache. We describe each in turn and defer the correctness arguments to
% \Cref{sec:theory}.

\subsection{Offline context-window encoding}
\label{sec:design:encoder}

\sysname preprocesses a user's memory once before serving, following Mem0's
extraction pipeline~\cite{chhikara2025mem0}: each conversation turn is distilled
into up to five salient \emph{facts}, and each fact becomes a memory chunk
encoded with the same LLM employed during inference. Simply encoding each fact
independently, however, loses the conversational context that helps disambiguate
it and leads to reduced retrieval quality after KV injection.

To address this, \sysname employs \emph{context-window encoding}. Each fact is
encoded together with a configurable window of the $w$ conversation turns
preceding the turn it was extracted from, but only the KV corresponding to the
fact is retained. This preserves local contextual information while allowing
every fact to remain independently retrievable and reusable.

To enable reuse at arbitrary prompt positions, the encoder stores keys before
rotary position encoding. We intercept the output of the key projection
($W_Kh$) before \texttt{apply\_rotary\_pos\_emb} and cache these unrotated keys,
together with the corresponding value tensors. The resulting reusable KV is
stored in GPU memory keyed by \texttt{fact\_id}. Because encoding is performed
offline, the encoder model is unloaded before serving begins and does not
consume GPU memory during inference.

% \sysname splits a user's memory into \emph{chunks} at the granularity of
% conversation turns and encodes each chunk once, offline, with the same model
% used for serving. To make a chunk relocatable to any serving-time position
% (\S\ref{sec:selective-retrieval}), the encoder captures keys \emph{before}
% rotary position embedding: we instrument \texttt{apply\_rotary\_pos\_emb} to
% intercept $k_{\text{pre}} = W_K h$ and store it unrotated, together with the
% (rotation-free) value tensor. Because a turn read in isolation loses the
% conversational context that disambiguates it, each chunk is encoded with a
% window of up to $w$ preceding sessions as a prefix, of which only the target
% turn's $K,V$ slice is retained. The encoded chunks reside on HBM keyed by
% \texttt{turn\_id}, and the encoder model is unloaded before the serving engine
% starts so that no encoder weights compete for GPU memory at serving time.

\subsection{Persistent memory store and retrieval}
\label{sec:design:retrieval}

Each fact is represented in two complementary forms linked by a shared
\texttt{fact\_id}: a \emph{semantic embedding} optimized for approximate
nearest-neighbor search, and a reusable \emph{KV representation} that conditions
the LLM during inference. This separation is deliberate: nearest-neighbor search
over raw attention keys is poorly behaved, whereas a dedicated embedding model
retrieves well, so each subsystem uses the representation it needs while the
\texttt{fact\_id} map keeps them referring to the same memory object.

Each fact's text is embedded (OpenAI \texttt{text-embedding-3-small}) and
indexed for similarity search. We manage this through
Mem0~\cite{chhikara2025mem0}, configured with \emph{Jasper}
(\S\ref{sec:bg:jasper}), a GPU-resident vector index, as its backend; the
embedding is stored under the same \texttt{fact\_id} as the fact's KV. At
serving time \sysname embeds the query with the same embedding model and issues
a top-$k$ search, yielding the \texttt{fact\_id}s of the most relevant facts.
Retrieval and injection thus operate in two different spaces, a contrastively
trained embedding space for search and the model's own KV space for
conditioning, linked only by the shared identifier.

\subsection{Chunked-RoPE composition}
\label{sec:design:compose}

The retrieved KV cannot be injected directly because it was encoded at a
different position than where it will appear in the serving prompt. Since RoPE
encodes token position directly into the keys, naïvely reusing stored KV would
produce incorrect attention scores.

Given the retrieved \texttt{fact\_id}s, \sysname gathers the corresponding
pre-RoPE facts, prepends a fixed instruction prefix, and composes them into a
single memory segment. Composition assigns the concatenated tokens contiguous
virtual positions $0,\ldots,m{-}1$ and applies RoPE to the stored keys at those
positions on the fly, using the model's own rotary tables.
\Cref{thm:chunked-rope} guarantees that relocating a fact's stored KV to a
chosen position yields exactly the attention scores that KV would receive if
prefilled at that position, so the composed KV is a valid drop-in for the paged
cache. The
query is placed immediately after the memory, at position $m$. Because facts
are encoded independently rather than as one joint prefill, the composed
segment omits the cross-chunk attention that a full prefill would include; we
quantify the resulting accuracy cost in \Cref{sec:eval}.

\subsection{KV injection via the vLLM KV connector}
\label{sec:design:connector}

Once composed, the retrieved memory is injected through vLLM's KV connector in
four steps.

\begin{enumerate}[leftmargin=*]
\item The connector registers the composed memory and its corresponding token
sequence.

\item During scheduling, vLLM recognizes that the initial prompt tokens are
already available and allocates paged KV blocks without scheduling them for
prefill.

\item Before execution, the connector copies the composed KV directly into those
blocks using GPU-to-GPU memory copies.

\item The model prefills only the query tokens, after which decoding proceeds
unchanged.
\end{enumerate}

The composed memory, its per-layer KV tensors and the token IDs that produced
them, is registered in a GPU-resident store, and the request prompt is formed
as $[\text{memory token IDs}\mid\text{query token IDs}]$. Our connector
implements vLLM's \texttt{KVConnectorBase\_V1} interface with two sides.
\textbf{Scheduler side:} \texttt{get\_num\_new\_matched\_tokens()} checks
whether the prompt begins with a registered memory token sequence; on a match
it reports those tokens as ``externally available,'' so vLLM allocates paged
blocks for them without scheduling them for prefill. \textbf{Worker side:}
before the forward pass, \texttt{start\_load\_kv()} scatter-copies the composed
KV into the allocated blocks via the slot mapping, a GPU-to-GPU copy on the
same device, after which vLLM prefills only the query tokens and attends to
the injected memory as if it had been prefilled. The connector is loaded as an
external plugin (\texttt{kv\_connector\_module\_path}) and touches no vLLM
internals, which keeps \sysname compatible with stock releases.

\subsection{Storage cost}
\label{sec:design:storage}
Each fact stores $2 \times L \times c \times h \times d \times b$ bytes of KV
(layers $L$, fact length $c$, KV heads $h$, head dimension $d$, byte width
$b$), plus one embedding vector for retrieval. For Llama-3.1-8B-Instruct with
GQA ($L{=}32$, $h{=}8$, $d{=}128$, bfloat16), $1{,}024$ memory tokens occupy
$\approx 128$\,MB on HBM, and the retrieval embeddings add a few kilobytes per
fact. Because facts are encoded once and reused across all of a user's
requests, this cost is amortized over the request stream.

% \para{Generality beyond memory} 
% Although motivated by personalized memory systems, \sysname applies to any
% workload that repeatedly conditions LLMs on largely static retrieved context,
% including retrieval-augmented generation, system prompts, few-shot exemplars,
% and tool specifications. More generally, the design demonstrates that
% conditioning can be performed directly at the attention layer rather than
% through repeated token-level prompt injection.

% Although we frame \sysname around
% personalized memory, the primitive---retrieving static context and injecting it
% as precomputed KV rather than prompt tokens---applies to any setting where a
% stable corpus is queried repeatedly: retrieval-augmented generation,
% system-prompt caching, few-shot exemplars, and tool-schema descriptions.
%
% ============================================================

% ============================================================

\section{Theoretical Analysis}
\label{sec:theory}

We now establish the correctness foundation for KV injection: the conditions
under which it produces \emph{exactly} the same outputs as prompt injection,
and the computational separation it induces. The base result assumes memory is
served at the same positions at which it was encoded. \Cref{sec:chunked-rope}
then lifts it to the position-reassignment case that selective retrieval
requires, and we are explicit below about which parts of \sysname the theorem
covers exactly and which it covers only approximately.

\subsection{Definitions and setup}
\label{sec:theory:setup}

\begin{definition}[Causal Decoder Transformer]
\label{def:transformer}
A causal decoder transformer $\mathcal{T}$ with $L$ layers maps an input token
sequence $(t_1, \ldots, t_n)$ to output hidden states via the recurrence:
\begin{align}
    \mathbf{h}_i^{(0)} &= \mathrm{Embed}(t_i), \quad i = 1, \ldots, n \\
    \mathbf{h}_i^{(\ell)} &= \mathcal{U}^{(\ell)}\!\Big(\mathbf{h}_i^{(\ell-1)},\; \{(\mathbf{k}_j^{(\ell)}, \mathbf{v}_j^{(\ell)})\}_{j \le i}\Big), \quad \ell = 1, \ldots, L
\end{align}
where each key/value pair $(\mathbf{k}_j^{(\ell)}, \mathbf{v}_j^{(\ell)})$ is a
(position-encoded) projection of $\mathbf{h}_j^{(\ell-1)}$, and the per-layer
update $\mathcal{U}^{(\ell)}$ subsumes causal self-attention over the preceding
key/value pairs together with the layer's residual connections, normalization,
and feed-forward transformation. Causal masking enforces that position $i$
depends on position $j$ only for $j \le i$. The analysis uses only two
properties of this recurrence: \emph{(i)} causality, and \emph{(ii)} that
preceding positions enter the update solely through their key/value pairs.
\end{definition}

\begin{definition}[Prompt Injection and KV Injection]
\label{def:methods}
Let $M = (t_1, \ldots, t_m)$ be a \emph{memory} token sequence and $Q =
(t_{m+1}, \ldots, t_{m+q})$ a \emph{query} token sequence.
\begin{itemize}[leftmargin=*]
    \item \textbf{Prompt injection (PI)} runs the full forward pass of
      $\mathcal{T}$ on the concatenation $[M; Q]$, producing hidden states
      $\mathbf{h}_i^{(\ell)}$ for all positions $i \in \{1, \ldots, m+q\}$ and
      layers $\ell \in \{0, \ldots, L\}$.
    \item \textbf{KV injection (KV)} first runs the forward pass of
      $\mathcal{T}$ on $M$ alone, producing KV pairs
      $(\tilde{\mathbf{k}}_j^{(\ell)}, \tilde{\mathbf{v}}_j^{(\ell)})$ for $j
      \in \{1, \ldots, m\}$ at each layer $\ell$. At serving time, the KV pairs
      are injected into the cache at positions $1, \ldots, m$, and the forward
      pass runs on $[M; Q]$ but \emph{skips} the computation of hidden states
      for positions $1, \ldots, m$, using the pre-computed KV pairs instead.
\end{itemize}
Both methods assign position index $j-1$ to token $t_j$ (i.e., positions $0,
\ldots, m-1$ for memory, $m, \ldots, m+q-1$ for query) for the purpose of
positional encoding.
\end{definition}

\subsection{Exact equivalence}
\label{sec:theory:equivalence}

\begin{theorem}[Exact Equivalence for Causal Decoder Transformers]
\label{thm:equivalence}
Let $\mathcal{T}$ be a causal decoder transformer
(Definition~\ref{def:transformer}) with $L$ layers, using any position encoding
scheme that assigns embeddings based solely on position index (including
RoPE~\citep{su2021roformer}, 
% ALiBi, or 
learned absolute embeddings). Let $M$
and $Q$ be memory and query sequences as in Definition~\ref{def:methods}, with
identical effective key/value tensors at each query position and deterministic
inference (i.e., no dropout, stochastic routing, training mode).
% with
% identical position assignments under both PI and KV.

Then for every layer $\ell \in \{1, \ldots, L\}$ and every query position $i
\in \{m+1, \ldots, m+q\}$:
\[
    \mathbf{h}_i^{(\ell)}\big|_{\mathrm{KV}} = \mathbf{h}_i^{(\ell)}\big|_{\mathrm{PI}}
\]
That is, the hidden states at all query positions are \textbf{identical} under
KV injection and prompt injection.
\end{theorem}

\begin{proof}
We proceed by induction on the layer index $\ell$.

\para{Base case ($\ell = 0$)} The embedding layer is position-wise:
$\mathbf{h}_i^{(0)} = \mathrm{Embed}(t_i)$ depends only on the token $t_i$,
which is the same under both methods. Thus
$\mathbf{h}_i^{(0)}\big|_{\mathrm{KV}} = \mathbf{h}_i^{(0)}\big|_{\mathrm{PI}}$
for all $i$.

\para{Inductive step} Assume that for layer $\ell - 1$, the hidden states agree
at all positions:
\[
    \mathbf{h}_j^{(\ell-1)}\big|_{\mathrm{KV}} = \mathbf{h}_j^{(\ell-1)}\big|_{\mathrm{PI}}, \quad \forall\; j \in \{1, \ldots, m+q\}.
\]

We must show the same holds at layer $\ell$. We consider memory positions ($j
\le m$) and query positions ($j > m$) separately.

\textbf{Memory positions} ($j \le m$). Under PI, $\mathbf{h}_j^{(\ell)}$ is
computed from $\{\mathbf{h}_k^{(\ell-1)}\}_{k \le j}$ via causal attention.
Since $j \le m$, the set $\{k : k \le j\}$ contains only memory positions.
Under KV injection, the forward pass on $M$ alone computes
$\tilde{\mathbf{h}}_j^{(\ell)}$ from $\{\tilde{\mathbf{h}}_k^{(\ell-1)}\}_{k
\le j}$, exactly the same set of positions (causal masking ensures no query
token is visible). By the inductive hypothesis,
$\tilde{\mathbf{h}}_k^{(\ell-1)} = \mathbf{h}_k^{(\ell-1)}\big|_{\mathrm{PI}}$
for all $k \le m$. Since $\mathcal{U}^{(\ell)}$ is a deterministic function of
its inputs, and the position indices are identical:
\[
    \tilde{\mathbf{h}}_j^{(\ell)} = \mathbf{h}_j^{(\ell)}\big|_{\mathrm{PI}}, \quad \forall\; j \le m.
\]
Consequently, the pre-computed KV pairs satisfy
\[
(\tilde{\mathbf{k}}_j^{(\ell)}, \tilde{\mathbf{v}}_j^{(\ell)}) =
(\mathbf{k}_j^{(\ell)}, \mathbf{v}_j^{(\ell)})\big|_{\mathrm{PI}}.
\]

\textbf{Query positions} ($i > m$). Under KV injection, the attention at
position $i$ uses:
\begin{itemize}[leftmargin=*]
    \item The \emph{injected} KV pairs $(\tilde{\mathbf{k}}_j^{(\ell)},
      \tilde{\mathbf{v}}_j^{(\ell)})$ for $j \le m$ (pre-computed from memory).
    \item The \emph{freshly computed} KV pairs $(\mathbf{k}_j^{(\ell)},
      \mathbf{v}_j^{(\ell)})$ for $m < j \le i$ (from the query forward pass).
\end{itemize}
From the memory-position argument above, the injected KV pairs are identical to
what PI would compute. From the inductive hypothesis, the query hidden states
$\mathbf{h}_j^{(\ell-1)}$ for $j > m$ are also identical. Therefore the layer
update at position $i$ receives identical inputs under both methods, and since
$\mathcal{U}^{(\ell)}$ is deterministic:
\[
\begin{aligned}
    \mathbf{h}_i^{(\ell)}\big|_{\mathrm{KV}}
      &= \mathcal{U}^{(\ell)}\!\big(\mathbf{h}_i^{(\ell-1)},\; \{(\mathbf{k}_j^{(\ell)}, \mathbf{v}_j^{(\ell)})\}_{j \le i}\big)\Big|_{\mathrm{KV}} \\
      &= \mathcal{U}^{(\ell)}\!\big(\mathbf{h}_i^{(\ell-1)},\; \{(\mathbf{k}_j^{(\ell)}, \mathbf{v}_j^{(\ell)})\}_{j \le i}\big)\Big|_{\mathrm{PI}} \\
      &= \mathbf{h}_i^{(\ell)}\big|_{\mathrm{PI}}.
\end{aligned}
\]

By induction, the claim holds for all layers $\ell \in \{1, \ldots, L\}$ and
all query positions $i \in \{m+1, \ldots, m+q\}$.
\end{proof}

\para{From the mechanism to the system} Theorem~\ref{thm:equivalence} makes two
assumptions that \sysname's retrieval pipeline does not literally satisfy, and
it is worth stating exactly how each is handled. First, it assumes memory is
injected at the same positions $0,\ldots,m{-}1$ at which it was encoded.
Selective retrieval instead composes chunks and places them at request-time
positions; \Cref{thm:chunked-rope} shows that storing keys \emph{pre-RoPE} and
re-rotating them at injection makes this relocation exact, so position
reassignment, which naively breaks the equivalence, is \emph{recovered}
rather than lost. Second, the theorem encodes the memory $M$ as a single block,
so memory tokens attend to one another. \sysname instead encodes each fact
independently, so a fact's KV omits the cross-fact left context that a joint
prefill would provide. KV injection is therefore exact
for a jointly encoded context and an \emph{approximation} for independently
encoded facts; the residual is the accuracy cost we measure in
\Cref{sec:eval}, and it is the price of encoding each fact once and reusing it
across a user's requests.

\begin{corollary}[Output Distribution Equivalence]
\label{cor:output}
Under the conditions of Theorem~\ref{thm:equivalence}, the next-token
probability distribution $P(t_{m+q+1} \mid t_1, \ldots, t_{m+q})$ is identical
under KV injection and prompt injection.
\end{corollary}

\begin{proof}
The output distribution is a deterministic function (the language model head)
of $\mathbf{h}_{m+q}^{(L)}$, which is identical under both methods by
Theorem~\ref{thm:equivalence}.
\end{proof}

\para{Remark} Theorem~\ref{thm:equivalence} provides an \emph{exact}
equivalence, not an approximation. This is a consequence of the causal
structure: memory tokens cannot attend to future query tokens, so their KV
representations are query-independent. The result holds regardless of the
specific attention variant (vanilla, multi-head, grouped-query, multi-query)
and regardless of the positional encoding scheme, provided position assignments
are consistent.

\subsection{Complexity separation}
\label{sec:theory:complexity}

\begin{proposition}[Complexity Separation]
\label{thm:complexity}
For a single attention layer with $n_h$ heads of dimension $d$, memory length
$m$, and query length $q$:
\begin{enumerate}[leftmargin=*]
    \item \textbf{Prompt injection} requires $\Theta\!\big((m+q)^2 \cdot n_h
      d\big)$ attention FLOPs for the prefill pass.
    \item \textbf{KV injection} requires $\Theta\!\big(q \cdot (m+q) \cdot n_h
      d\big)$ attention FLOPs at serving time (excluding the one-time encoding
      cost).
    \item The per-request FLOP savings is:
    \[
        \Delta = \Theta(m^2 \cdot n_h d + m \cdot q \cdot n_h d) = \Theta(m^2 \cdot n_h d) \quad \text{when } m \gg q.
    \]
\end{enumerate}
The total savings across $L$ layers and $R$ requests is $\Theta(R \cdot L \cdot
m^2 \cdot n_h d)$, while the one-time encoding cost is $\Theta(L \cdot m^2
\cdot n_h d)$---amortized to zero as $R \to \infty$. Under compute-bound
prefill, TTFT is proportional to attention FLOPs.
\end{proposition}

\para{Interpretation under retrieval} In \sysname the injected memory is the
concatenation of the $k$ retrieved chunks, so $m = \sum_i |C_i|$ grows with the
retrieval budget $k$. Under prompt injection the per-request prefill is
$\Theta((m+q)^2)$ and thus grows with $k$; under KV injection it is
$\Theta(q(m+q))$, grows linearly rather than quadratically with $k$.
% dominated by the fixed query length $q$ when $q \ll m$.
KV injection eliminates the quadratic memory-prefill term. Its remaining online
attention cost grows linearly with retrieved memory length.
% Time-to-first-token is therefore nearly independent of how much memory is
% retrieved, the behavior we observe empirically in~\Cref{sec:eval}.

\section{Chunked RoPE}
\label{sec:chunked-rope}

\begin{figure}[tbp]
  \centering
  \includegraphics[width=\columnwidth]{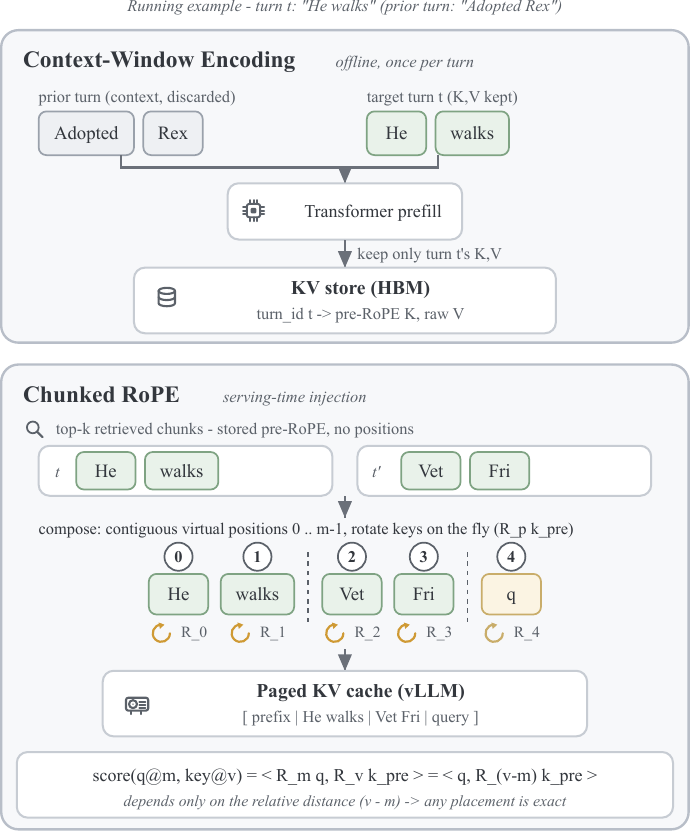}
  \caption{The two techniques on one minimal example (memory turn $t$:
  ``He walks,'' with prior turn ``Adopted Rex'').
  \textbf{Top: context-window encoding (offline).} Turn $t$ is encoded behind a
  window of preceding turns in a single forward pass, but only $t$'s pre-RoPE
  $K$ and raw $V$ are kept: the stored slice can condition `He` on `Rex`,
  without adding the context tokens to each request.
  % already resolves ``He'' to Rex, 
  % at
  % no serving-time cost. 
  \textbf{Bottom: chunked RoPE (serving).} Retrieved
  chunks are stored without positions ($k_{\text{pre}}=W_K h$); at injection
  they receive contiguous virtual positions $0..m{-}1$ and their keys are
  rotated on the fly. Because the score
  $\langle R_m q, R_v k_{\text{pre}}\rangle = \langle q, R_{v-m}
  k_{\text{pre}}\rangle$ depends only on the relative distance $v-m$, any
  placement is exact.}
  \label{fig:rope-encoding}
\end{figure}

The equivalence theorem assumes every memory chunk is encoded at the same
position where it will later be served. However, selective retrieval violates
this assumption because a chunk may appear at different prompt locations for
different requests. In this section, we present \emph{chunked RoPE} wherein we
store the keys before applying RoPE and remove this dependency. 
Relocating stored keys is exact for fixed hidden states, but independently
encoding facts omits left-context interactions and can reduce answer quality.
We address this separate source of approximation with \emph{context-window encoding}.
% Although, this
% position independence results in a nominal accuracy loss we later show that we
% can recover it via \emph{context-window encoding}.

% Selective retrieval is a prerequisite for scaling \sysname beyond a fixed
% working set: at request time the system must pick the top-$k$ relevant chunks
% from a memory store of arbitrary size and inject them into paged attention.
% This raises two questions that are easy to conflate. First, must the position
% assigned to a chunk be known at \emph{encoding} time? Second, given an
% injection position, how much quality is lost relative to a joint prefill of the
% same chunks? \Cref{thm:equivalence} answers neither: it assumes memory is
% encoded at exactly the positions at which it is served. Chunked RoPE removes
% that assumption, turning position reassignment from a case that breaks the
% equivalence into one it handles exactly.

\subsection{Position-agnostic storage}
\label{sec:selective-retrieval}
\sysname stores keys in \emph{pre-RoPE} form, intercepting $k_{\text{pre}} =
W_K h$ before \texttt{apply\_rotary\_pos\_emb}. The following theorem shows that
the virtual position assigned to a chunk at injection is then independent of
any choice made at encoding. A chunk stored once may be injected by different
requests at different virtual positions $v$, with attention scores given
exactly by $\langle R_v^\top q,\, k_{\text{pre}}\rangle$ for the pre-rotated
query (\Cref{fig:rope-encoding}, bottom).

\begin{theorem}[Chunked RoPE Injection]
\label{thm:chunked-rope}
Let $\mathcal{M}$ be a transformer with rotary position embedding (RoPE),
trained on sequences with relative positions in $[-D_{\max}, D_{\max}]$.
Let $\mathcal{C} = \{(k_{p_0+j}, v_{p_0+j}) : 0 \le j < c\}$ denote a
contiguous span of $c$ tokens from the source context, stored in
\emph{pre-RoPE} form (i.e., $k_p \coloneqq W_K h_p$ before any positional
rotation). Suppose $\mathcal{C}$ is virtually injected at offsets
$[v_0, v_0+c)$ and queried by $q_N$ at virtual position $N$. Then:

\begin{enumerate}[leftmargin=*]
  \item \textbf{Shift equivalence.} The attention score between $q_N$ and
    the injected key at virtual position $v_0 + j$ equals the score that
    $\mathcal{M}$ would produce for the same key at its original position
    $p_0 + j$ when queried by $q_{N'}$ at
    \[
      N' \;\coloneqq\; N - (v_0 - p_0).
    \]
    That is, $\langle R_N q,\, R_{v_0+j}\, k_{p_0+j}\rangle
        = \langle R_{N'} q,\, R_{p_0+j}\, k_{p_0+j}\rangle$ for all $j$.
  \item \textbf{Within-chunk preservation.} The attention pattern over
    $\mathcal{C}$ at virtual offset $v_0$ is structurally identical to a
    contiguous span seen at training time, provided the per-token relative
    distances $\{N - (v_0+j) : 0\le j<c\}$ lie within $[-D_{\max}, D_{\max}]$.
    No intra-chunk relative distance is altered.
  \item \textbf{Cross-chunk independence.} For $K$ injected chunks
    $\mathcal{C}_1,\dots,\mathcal{C}_K$ at offsets $v_1,\dots,v_K$, the
    cross-attention scores depend only on the per-token relative distances
    $\{N - v : v \in \bigcup_i \mathcal{C}_i\}$. Inter-chunk virtual gaps
    $v_i - v_j$ enter only through the softmax normalization, never through
    a key--key RoPE interaction. Chunked RoPE introduces no additional
    interaction beyond standard attention normalization.
\end{enumerate}
\end{theorem}

\begin{proof}[Proof sketch]
RoPE acts as a block-diagonal rotation $R_p$ with frequencies
$\{\theta_i\}_{i=1}^{d/2}$, and satisfies $R_a^\top R_b = R_{b-a}$.
Hence $\langle R_a q,\, R_b k\rangle = \langle q,\, R_{b-a} k\rangle$,
i.e., scores depend on $b-a$ alone. For (1), the chunked score is
$\langle q,\, R_{(v_0+j)-N}\, k\rangle$; substituting $N' = N - (v_0-p_0)$
gives $(v_0+j)-N = (p_0+j)-N'$, recovering the full-context score.
(2) follows because $\{N-(v_0+j)\}_{j=0}^{c-1}$ is an arithmetic
progression with unit step, identical in form to any training-time chunk.
(3) follows because cross-attention from $q_N$ to memory keys involves
no key--key RoPE interaction; the softmax denominator
$\sum_v \exp(\cdot/\sqrt{d})$ couples chunks only through individual
$\{N-v\}$ terms.
\end{proof}

\begin{corollary}[In-distribution injection budget]
\label{cor:budget}
For any retrieval scheme, chunked injection remains in-distribution iff
every injected token's virtual position $v$ satisfies $|N - v| \le D_{\max}$.
Equivalently, the total budget of virtual positions consumable by retrieved
chunks is at most $\min(N, D_{\max})$, occupying positions to the left of the
query.
\end{corollary}

A fact therefore carries no positional metadata in storage: only its content
embedding (for retrieval) and its raw $K,V$ tensors are persisted, and RoPE is
applied at injection for whatever positions the request assigns. This is
exactly what lets~\Cref{sec:design:compose} compose independently encoded
facts at request time.

\subsection{Context-window encoding}
\label{sec:context-window-encoding}

% Position-agnostic storage settles \emph{where} a chunk may be injected; it says
% nothing about \emph{what} the chunk stores. 

A fact read in isolation is often underspecified, e.g., a pronoun,
an ellipsis, or a deictic reference (``he,'' ``there,'' ``the same one'')
resolves only against earlier turns. Encoding a fact without cross-turn
attention would bake that ambiguity into its $K,V$, and no choice of
injection position at serving time can recover a referent the keys never saw.

\sysname therefore encodes each fact \emph{in context}. To capture a fact
extracted from turn $t$, the encoder prepends a window of up to $w$ turns
preceding $t$ and runs a single forward pass over the concatenation, but
retains only the $K,V$ slice belonging to the fact itself
(\Cref{fig:rope-encoding}, top), the prefix is discarded. Because keys are
intercepted before rotation ($k_{\text{pre}} = W_K h$,
\Cref{sec:selective-retrieval}), the retained slice is still position-agnostic
and composes under chunked RoPE exactly as~\Cref{thm:chunked-rope} requires.
The prefix changes only the \emph{content} of the fact's hidden states, which
now encode the disambiguating context, not the positions at which the fact can
later be served.

The stored slice is thus simultaneously (1)~\emph{context-aware}, i.e., the
fact's keys already reflect the referents established by its prefix (in
\Cref{fig:rope-encoding}, ``He'' resolves to Rex), and (2)~\emph{relocatable}
to any serving-time position. The cost of context-window encoding is paid once,
offline, at serving time the prefix is neither retrieved nor prefilled, so a
fact carries its context for free.

% Context-window encoding and chunked RoPE
% are therefore complementary halves of the same mechanism, the former fixes what
% each chunk \emph{means}, the latter fixes where it \emph{goes}, and together
% they let independently encoded turns be composed at request time without loss.

Together, chunked RoPE and context-window encoding decouple where a fact is
served from how it is encoded. The former guarantees exact position
relocation, while the latter preserves the contextual information needed for
accurate conditioning. This separation enables \sysname to retrieve, compose,
and inject facts independently at serving time without re-encoding.

\begin{theorem}[Context-Window Encoding]
\label{thm:context-window}
Let $P=(x_1,\ldots,x_w)$ denote a prefix consisting of the $w$ conversation
turns immediately preceding a target fact $T=(x_{w+1},\ldots,x_{w+c})$ (the $c$
tokens of the fact). Suppose the transformer is executed once on the
concatenated sequence $P \circ T$. Let $K_T,\;V_T$ denote the key and value
tensors corresponding only to the tokens of $T$.
Then:
\begin{enumerate}[leftmargin=*]
\item The stored $(K_T,V_T)$ are exactly the tensors that would have been
produced for $T$ during a standard prefill over $P\circ T$.

\item Every token in $T$ incorporates information from every token in $P$
through the transformer's causal attention.

\item No information from tokens occurring after $T$ contributes to
$(K_T,V_T)$.
\end{enumerate}

Consequently, context-window encoding stores the exact representation of the
target fact conditioned on the chosen encoding window.
\end{theorem}

\begin{proof}
Consider the standard autoregressive forward pass over the sequence $P\circ T$.
Under causal attention, the hidden state of every token depends only on earlier
tokens in the sequence.
Since every token of $P$ precedes every token of $T$, the hidden states
computed for tokens in $T$ incorporate all information contained in the prefix
$P$ exactly as they would during ordinary inference.

Because keys and values are deterministic projections of these hidden states,
the tensors $(K_T,V_T)$ extracted from the forward pass are exactly those that
would appear in the KV cache after prefilling $P\circ T$.
Finally, causal masking prevents every token occurring after the end of $T$
from influencing the hidden states of tokens within $T$. Therefore the computed
$(K_T,V_T)$ depend only on $P$ and $T$, establishing (1)--(3).
\end{proof}

\begin{corollary}[Exactness under Prefix Equality]
\label{cor:window}
Suppose the token sequence preceding the target fact $T$ in its encoding window
is \emph{identical, as a sequence}, to the sequence preceding $T$ in the serving
prompt. Then the stored $(K_T,V_T)$ are identical to those produced by a joint
prefill of the serving prompt. If the encoding prefix differs in any token,
including by containing additional context absent at serving time, equality
does not hold in general.
\end{corollary}

\noindent The corollary follows from \Cref{thm:context-window}(1): under causal
attention, $(K_T,V_T)$ are a function of $T$ and its preceding tokens only, so
equal prefixes yield equal tensors. The converse direction is equally
important. Containing the retrieved facts is \emph{not} sufficient: any extra
token in the encoding window, such as conversation turns that are not retrieved
at serving time, changes the hidden states of $T$ and hence $(K_T,V_T)$.

Because retrieval is query-dependent, the serving-time prefix of a fact varies
from request to request, and exact prefix equality will generally not hold in
deployment. Context-window encoding is therefore a systematic
\emph{approximation}, not an exactness guarantee. It replaces the unknown,
query-dependent serving prefix with the fact's fixed source context, restoring
the referents and conversational context that isolated encoding loses. The
residual mismatch between the two prefixes is what the accuracy gap to prompt
injection in \Cref{sec:eval:accuracy} measures.

\para{Position policy}
% \label{sec:position-policy}
We adopt the simplest assignment consistent with the theorem. The retrieved
facts $C_1,\ldots,C_k$ are placed contiguously starting at virtual position
zero, so fact $C_i$ begins at $b_i = \sum_{j<i} |C_j|$ and the query follows
immediately after the memory at $b_q = m = \sum_i |C_i|$. (Facts are
variable-length, so the offsets are cumulative lengths rather than multiples of
a fixed chunk size.) This preserves each fact's internal relative positions
exactly (Theorem~\ref{thm:chunked-rope}(2)) while discarding the original
inter-fact gaps. Furthermore, we order the retrieved facts in reverse order of
proximity and position the most relevant fact adjacent to the query.

\section{Evaluation}
\label{sec:eval}

We now evaluate \sysname as a \emph{memory system for LLMs}. In our evaluation, we compare
\sysname against the state-of-the-art production memory system
Mem0~\cite{chhikara2025mem0}. 
% the standard serving mechanism, \emph{prompt injection}, used by production
% memory systems such as Mem0~\cite{chhikara2025mem0}. 
\sysname contributes two techniques over a Mem0-style memory pipeline:
\textbf{GPU-native retrieval}, which serves the vector index from GPU memory
instead of a CPU vector store, and \textbf{KV Injection (KV)}, which conditions
on retrieved facts by inserting their precomputed KV tensors directly into the cache
rather than re-prefilling them as prompt text. Mem0, by contrast, uses a CPU
vector store and \textbf{Prompt Injection (PI)}, serializing the retrieved facts
into the input prompt and recomputing their KV on every request. Both systems
extract and retrieve the same top-$k$ facts using Mem0's pipeline. Unless noted
otherwise, \textbf{Mem0} denotes its default production configuration (a Qdrant
CPU vector store with prompt injection), so that \sysname-vs-Mem0 comparisons
reflect the combined effect of both contributions; the vector-backend ablation
(\Cref{tab:vector-backend}) isolates the GPU-native-retrieval contribution on
its own.
Our evaluation answers four questions:

\begin{itemize}[leftmargin=*]
  \item \textbf{RQ1 (serving latency).} Is serving latency nearly independent of the
    amount of retrieved memory, and what is the GPU-memory cost of the KV store?
    % (\Cref{sec:eval:latency,sec:eval:footprint})
  \item \textbf{RQ2 (accuracy).} Does context-window encoding recover the
    accuracy lost to independent chunk encoding?
    % (\Cref{sec:eval:accuracy})
  \item \textbf{RQ3 (serving throughput).} How does serving throughput scale
    with concurrent users?
  \item \textbf{RQ4 (CPU offloading).} How does offloading the pre-computed KV
    embeddings to CPU impact the serving latency?
\end{itemize}

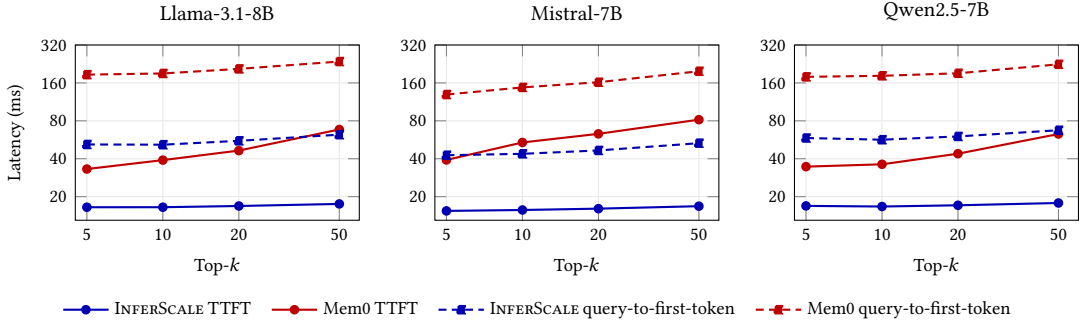
\begin{figure*}[t]
\centering
\begin{tikzpicture}
\begin{groupplot}[lataxis, group style={group size=3 by 1, horizontal sep=1.0cm}]
\nextgroupplot[title={Llama-3.1-8B}, ylabel={Latency (ms)}]
\addplot[isttft]  coordinates {(5,16.49)(10,16.50)(20,16.87)(50,17.52)};
\addplot[memttft] coordinates {(5,33.15)(10,39.03)(20,46.44)(50,68.33)};
\addplot[isq2ft]  coordinates {(5,51.94)(10,51.72)(20,55.52)(50,62.12)};
\addplot[memq2ft] coordinates {(5,186.15)(10,190.22)(20,206.79)(50,236.48)};
\nextgroupplot[title={Mistral-7B}]
\addplot[isttft]  coordinates {(5,15.42)(10,15.68)(20,16.07)(50,16.81)};
\addplot[memttft] coordinates {(5,39.15)(10,53.88)(20,63.08)(50,81.71)};
\addplot[isq2ft]  coordinates {(5,42.66)(10,43.77)(20,46.58)(50,53.18)};
\addplot[memq2ft] coordinates {(5,129.37)(10,147.41)(20,162.01)(50,198.07)};
\nextgroupplot[title={Qwen2.5-7B}]
\addplot[isttft]  coordinates {(5,16.90)(10,16.71)(20,17.09)(50,17.82)};
\addplot[memttft] coordinates {(5,34.64)(10,36.15)(20,43.87)(50,62.89)};
\addplot[isq2ft]  coordinates {(5,58.46)(10,56.66)(20,60.16)(50,67.37)};
\addplot[memq2ft] coordinates {(5,178.81)(10,182.25)(20,190.83)(50,224.76)};
\end{groupplot}
\end{tikzpicture}

\vspace{2pt}
\begin{tikzpicture}
\begin{axis}[hide axis, scale only axis, height=2mm, width=0.85\textwidth,
  xmin=0, xmax=1, ymin=0, ymax=1, legend columns=4,
  legend style={draw=none, at={(0.5,0.5)}, anchor=center, font=\footnotesize,
  /tikz/every even column/.append style={column sep=0.6em}}]
\addlegendimage{isttft}\addlegendentry{\sysname{} TTFT}
\addlegendimage{memttft}\addlegendentry{Mem0 TTFT}
\addlegendimage{isq2ft}\addlegendentry{\sysname{} query-to-first-token}
\addlegendimage{memq2ft}\addlegendentry{Mem0 query-to-first-token}
\end{axis}
\end{tikzpicture}
\caption{Serving latency on \locomo (log scale, ms) vs.\ retrieval budget $k$
across three models: vLLM engine TTFT (solid) and end-to-end query-to-first-token
(dashed) for \sysname and Mem0 (default configuration), with prefix caching
enabled. \sysname's latency is nearly flat in $k$ and independent of the
encoding window (curves for $w\in\{0,5,20,50\}$ overlap within $1$\,ms; a
representative $w{=}5$ is shown), whereas Mem0 grows with $k$.}
\label{fig:latency}
\end{figure*}

% ------------------------------------------------------------
\subsection{Experimental setup}
\label{sec:eval:setup}

\para{Hardware and software}
All experiments run on a single NVIDIA RTX PRO 6000 Blackwell Server
Edition with 96GB of GDDR7 memory and peak bandwidth of 1597GB/s. The GPU is
connected through PCI Express 5.0 x16 with peak read throughput of 64GB/s.
\sysname is implemented as a KV connector plugin in vLLM~v0.19.1. The KV cache
uses \texttt{bfloat16} precision with a 16-token page size. 
We enable vLLM's \emph{prefix caching} in all benchmarks so that prompt injection can
receive any available speed ups if input tokens are repeated.
% no cross-request prefix
% reuse masks the cost we study: prompt injection therefore pays the full prefill
% of its retrieved memory on every request, matching the behavior a deployed
% memory system exhibits when the retrieved subset and its ordering vary from
% query to query (we measure the prefix-caching best case separately in
% \Cref{sec:eval:position}).
%

\para{Models}
We evaluate three open-weight instruction-tuned models:
\textbf{Llama-3.1-8B-Instruct}~\cite{llama3} (our primary target),
\textbf{Mistral-7B-Instruct-v0.3}~\cite{mistral7b}, and
\textbf{Qwen2.5-7B-Instruct}~\cite{qwen25}. They span
different attention and positional-encoding configurations (e.g., grouped-query
head counts~\cite{ainslie2023gqa} and RoPE parameterizations), testing that \sysname's mechanism is
model-agnostic rather than tuned to one architecture. \sysname requires no
fine-tuning; each model is used as released. 
To evaluate the robustness of \sysname on higher-parameter models, we also
evaluate \textbf{Qwen3-14B}~\cite{qwen3} (\Cref{sec:eval:scale}).

\para{Dataset}
Our workload is
long-conversation question answering on \locomo~\cite{locomo}. We use the
ten-conversation subset, which after excluding the adversarial/unanswerable
category comprises 1{,}540 answerable QA pairs across four categories: 282
multi-hop, 321 temporal, 96 open-domain, and 841 single-hop. Following Mem0's
preprocessing pipeline~\cite{chhikara2025mem0}, we extract up to five salient
facts from each dialogue turn; each fact is stored as one memory chunk, and its
text is embedded for retrieval. All accuracy numbers are micro-averaged over
the 1{,}540 answerable questions. We precompute 1{,}536-dimensional
\texttt{text-embedding-3-small}~\cite{neelakantan2022embeddings} embeddings for both memory facts and questions
and cache them on disk, so the measured request path performs a cached vector
lookup rather than a live embedding call. The Jasper index~\cite{mccoy2026gpu}
uses inner-product distance with 64 graph neighbors and beam width 64.
All one-time preprocessing—fact extraction, embedding generation, per-fact KV
encoding, Jasper index construction, and vLLM startup, is treated as an offline
setup cost and excluded from all reported latency and throughput.
% the measured request path
% begins at query submission.

\para{Conditions}
We compare two end-to-end systems.
% labeled \sysname and Mem0 throughout our results. 
\textbf{Mem0}~\cite{chhikara2025mem0}, the production baseline, serializes the
top-$k$ retrieved facts into the prompt and re-prefills them on every request
(\emph{prompt injection}). Mem0 uses Qdrant~\cite{qdrant} as the vector index for storing
and retrieving memory facts. \textbf{\sysname} retrieves from a GPU-resident
index and injects the retrieved facts as pre-RoPE
KV composed with chunked RoPE (\emph{KV injection}). We evaluate \sysname in two
configurations that differ only in where its pre-RoPE KV store resides:
\emph{GPU-resident} (the default) and \emph{CPU-offloaded}, which keeps the store
in host DRAM and streams the retrieved KV to the GPU over PCIe to scale capacity
beyond GPU memory (\S\ref{sec:eval:footprint},~\S\ref{sec:eval:cpu-offload}); we
report both where the distinction matters. We sweep the retrieval
budget $k\in\{5,10,20,50\}$ and, for \sysname, the offline encoding window
$w\in\{0,5,20,50\}$ (measured in complete conversation turns preceding a fact's
source turn; $w{=}0$ encodes each fact in isolation). Note that increasing $w$
changes only the information captured during offline encoding and its one-time
preprocessing cost. It does not affect serving-time retrieval, KV injection, or
GPU memory consumption.
Both systems retrieve with Mem0's pipeline using the same embeddings and budget
$k$; the GPU and CPU indices return comparable facts (\Cref{tab:vector-backend}).
% only the
% injection path differs.

\para{Metrics}
We evaluate \sysname across three axes. \emph{Accuracy} is the fraction of
answerable questions judged correct by an independent LLM judge,
Gemma-2-9B-Instruct~\cite{gemma2024}, prompted zero-shot for a binary verdict
against the reference answer (a refusal on an answerable question counts as
incorrect).
% ; an OpenAI judge agrees with Gemma in ordering, at lower absolute
% scores. % TODO: restore "\Cref{app:app-results}" pointer once the appendix
% % (application-results.tex) is re-enabled in main.tex.
For latency, we report two quantities.
\textbf{TTFT} is the vLLM engine's time to first token. It \emph{excludes}
retrieval and KV composition and isolates what injection changes inside the engine. 
\textbf{Query-to-first-token} is the end-to-end latency and
additionally includes the cached question-vector lookup, ANN search, KV
composition and registration, and prompt assembly. 
% We use TTFT for the
% mechanism-level comparison and query-to-first-token where end-to-end cost
% matters (\Cref{sec:eval:latency}).
%
Finally, we report \textbf{throughput} as the number of queries served per
second (QPS) under concurrent load. Starting from a single user, we increase
the number of users from 10 to 100 (10, 25, 50, 75, 100) issuing \locomo
queries simultaneously and record the sustained QPS at each concurrency level. 
% scaling up until the additional KV state no
% longer fits in GPU memory. The peak QPS therefore reflects the maximum number of
% concurrent users a single device can sustain before its memory capacity is
% exhausted.
% \todo{Verify if we are reporting any other quantity.}
% Where retrieval quality is at issue we also report recall@$k$.

% ------------------------------------------------------------
\subsection{Serving latency}
\label{sec:eval:latency}

We report two latencies: engine \textbf{TTFT}, which isolates what injection
changes inside vLLM, and end-to-end \textbf{query-to-first-token}, which
additionally includes retrieval, KV composition, and prompt assembly.
\Cref{fig:latency} plots both for \sysname and Mem0.
Prefix caching is enabled, giving Mem0 its best case. The full per-window
breakdown is in \Cref{tab:latency-full} of \Cref{app:latency}. Note that \sysname's latency is nearly
flat in $k$ and independent of the encoding window (curves for
$w\in\{0,5,20,50\}$ overlap within $1$\,ms; a representative $w{=}5$ is shown).

\para{TTFT} \sysname's engine TTFT is nearly flat with increasing top-$k$
($\approx$16--18\,ms across $k$ on every model), because retrieved memory is
injected rather than prefilled. Mem0 re-prefills the retrieved facts, so its
TTFT grows with $k$, reaching $68$/$82$/$63$\,ms at $k{=}50$ on
Llama/Mistral/Qwen, $3.6$--$4.8\times$ \sysname's. TTFT is also independent of
the encoding window (\sysname's $w$ curves overlap within $1$\,ms), confirming
that $w$ is an offline-only cost.

\para{End-to-end} The same performance gap holds end-to-end, and the gap
widens. \sysname's query-to-first-token rises only slightly with $k$ (e.g.,
$52\to60$\,ms on Llama) because retrieval (Jasper) and injection both run on
the GPU. Mem0 additionally pays CPU-side vector search and host-to-device
transfer, so its query-to-first-token is $3.4$--$3.9\times$ higher and grows
faster, $236$/$198$/$225$\,ms at $k{=}50$ versus \sysname's
$\approx$$60$/$53$/$66$\,ms.

\para{Vector backend} The retrieval index determines chiefly \emph{where}
nearest-neighbor search runs. The two backends return comparable top-$k$ sets
with only minor differences in accuracy; the backend's large effect is on
latency.
\Cref{tab:vector-backend} runs Mem0 (prompt injection) with either the GPU-native
Jasper index or a conventional CPU store (Qdrant). The CPU index adds a large,
roughly fixed cost due to host-side ANN search plus a host-to-device transfer
of the retrieved text, inflating query-to-first-token by $100$--$150$\,ms
(e.g., $187$ vs.\ $48$\,ms at $k{=}5$ on Llama). This is precisely the overhead
the Mem0 baseline in~\Cref{fig:latency} pays; \sysname avoids it by keeping
retrieval on the same device as the KV store and the engine.

\begin{table}[t]
  \centering
  \small
  \caption{Retrieval backend: end-to-end query-to-first-token (ms) for Mem0
    with the GPU-native Jasper index vs.\ a conventional CPU store
  (Qdrant). The two backends return comparable results; the backend's large
  effect is on retrieval latency.}
  \label{tab:vector-backend}
  \begin{tabular}{@{}ll rrrr@{}}
    \toprule
    Model & Backend & $k{=}5$ & $10$ & $20$ & $50$ \\
    \midrule
    \multirow{2}{*}{Llama-3.1-8B} & Jasper (GPU) & 47.92 & 53.13 & 60.86 & 86.56 \\
      & Qdrant (CPU) & 186.15 & 190.22 & 206.79 & 236.48 \\
    \midrule
    \multirow{2}{*}{Mistral-7B} & Jasper (GPU) & 39.72 & 45.78 & 56.56 & 88.41 \\
      & Qdrant (CPU) & 129.37 & 147.41 & 162.01 & 198.07 \\
    \midrule
    \multirow{2}{*}{Qwen2.5-7B} & Jasper (GPU) & 51.30 & 54.87 & 63.36 & 85.24 \\
      & Qdrant (CPU) & 178.81 & 182.25 & 190.83 & 224.76 \\
    \bottomrule
  \end{tabular}
\end{table}

% ------------------------------------------------------------
\subsection{Scaling the retrieval budget}
\label{sec:eval:scaling-k}

\begin{table}[t]
  \centering
  \small
  \setlength{\tabcolsep}{4pt}
  \caption{Scaling the retrieval budget on \locomo (Llama-3.1-8B). Engine TTFT
  (ms) for prompt injection (PI, Mem0) and \sysname's KV injection, the
  resulting TTFT and end-to-end query-to-first-token speedups, and the share of
  the total query-to-first-token saving that occurs inside the engine. The
  $k{=}50$ column is from \Cref{fig:latency}.}
  \label{tab:scaling-k}
  \begin{tabular}{@{}lrrrr@{}}
    \toprule
    & $k{=}50$ & $100$ & $200$ & $400$ \\
    \midrule
    TTFT, PI (ms)                 & 68.33 & 89.2 & 162.6 & 343.2 \\
    TTFT, \sysname (ms)           & 17.52 & 17.7 & 20.0  & 24.7 \\
    TTFT speedup                  & $3.9\times$ & $5.0\times$ & $8.1\times$ & $13.9\times$ \\
    Q-to-first-tok.\ speedup      & $3.8\times$ & $4.4\times$ & $4.2\times$ & $4.2\times$ \\
    Engine share of saving        & $29\%$ & $30\%$ & $45\%$ & $62\%$ \\
    \bottomrule
  \end{tabular}
\end{table}

\sysname's end-to-end gain over Mem0 comes from two sources: GPU-native
retrieval and KV injection. At the operating points above ($k{\le}50$), most of
the query-to-first-token saving comes from retrieval. For Llama-3.1-8B at
$k{=}50$, $50.8$\,ms of the $174.4$\,ms saving ($29\%$) happens inside the
engine, and the rest comes from avoiding CPU-side search
(\Cref{tab:vector-backend}). The two mechanisms scale differently, however.
Retrieval cost depends mainly on $k$. Prompt injection, by contrast, re-prefills
every retrieved token on each request. To separate the two, we extend the sweep
to $k\in\{100,200,400\}$ (\Cref{tab:scaling-k}).

Engine TTFT isolates KV injection, because both arms consume the same retrieved
text and differ only in how it reaches the model. Prompt injection's TTFT grows
from $89$\,ms at $k{=}100$ to $343$\,ms at $k{=}400$. Each doubling of $k$ adds
$73.4$ and then $180.6$\,ms. \sysname's TTFT stays between $18$ and $25$\,ms,
adding only $2.3$ and then $4.7$\,ms per doubling. The engine-level speedup
therefore grows from $5.0\times$ to $13.9\times$.
We conservatively attribute \emph{all} of the non-engine saving to retrieval.
Even so, the engine share of the total query-to-first-token saving rises from
$30\%$ at $k{=}100$ to $62\%$ at $k{=}400$, crossing half between $k{=}200$ and
$400$. At $k{=}400$, KV injection saves $318.5$\,ms against $191.7$\,ms for
everything else. As the retrieved context grows, KV injection rather than
GPU-native retrieval becomes the dominant source of \sysname's speedup.

The end-to-end ratio stays near $4.2\times$ because \sysname's own retrieval
and KV-composition cost also grows with $k$. The engine-side advantage keeps
widening; the end-to-end advantage is limited by the composition path. The
document-RAG benchmarks (\Cref{app:rag-benchmarks}) show the same effect along
a different axis. There, each retrieved unit is a $1024$-token chunk, so even
$k{=}15$ yields ${\sim}15$K retrieved tokens. With the Jasper index and
retrieved chunks held fixed, KV injection lowers TTFT by $21.7$--$25.2\times$.

% ------------------------------------------------------------
\subsection{End-to-end memory QA accuracy}
\label{sec:eval:accuracy}
\begin{table}[t]
  \centering
  \small
  \caption{Judged \locomo accuracy (\%) vs.\ retrieval budget $k$ for \sysname
  (per encoding window $w$) and Mem0 (prompt injection), micro-averaged over the
  answerable questions. Per-category results are in
  \Cref{app:accuracy}.}
  \label{tab:accuracy}
  \begin{tabular}{@{}ll rrrr@{}}
    \toprule
    Model & Method & $k{=}5$ & $10$ & $20$ & $50$ \\
    \midrule
    \multirow{5}{*}{Llama-3.1-8B}
      & \sysname ($w{=}0$)  & 62.21 & 62.14 & 57.66 & 53.77 \\
      & \sysname ($w{=}5$)  & 60.00 & 59.81 & 59.87 & 56.82 \\
      & \sysname ($w{=}20$) & 60.13 & 62.08 & 61.62 & 59.22 \\
      & \sysname ($w{=}50$) & 60.06 & 61.36 & 62.66 & 60.26 \\
      & Mem0                & 56.95 & 59.29 & 61.49 & 63.25 \\
    \midrule
    \multirow{5}{*}{Mistral-7B}
      & \sysname ($w{=}0$)  & 54.03 & 55.13 & 52.84 & 37.22 \\
      & \sysname ($w{=}5$)  & 56.04 & 57.60 & 58.09 & 56.45 \\
      & \sysname ($w{=}20$) & 53.96 & 55.97 & 56.95 & 58.24 \\
      & \sysname ($w{=}50$) & 55.78 & 56.49 & 58.70 & 58.30 \\
      & Mem0                & 61.30 & 63.96 & 65.00 & 64.35 \\
    \midrule
    \multirow{5}{*}{Qwen2.5-7B}
      & \sysname ($w{=}0$)  & 59.74 & 56.75 & 54.22 & 50.45 \\
      & \sysname ($w{=}5$)  & 59.29 & 57.53 & 58.05 & 57.40 \\
      & \sysname ($w{=}20$) & 60.00 & 59.68 & 58.25 & 58.64 \\
      & \sysname ($w{=}50$) & 58.83 & 58.25 & 58.18 & 58.77 \\
      & Mem0                & 60.06 & 63.64 & 65.13 & 64.61 \\
    \bottomrule
  \end{tabular}
\end{table}

\Cref{tab:accuracy} reports overall judged accuracy for \sysname (per window
$w$) and Mem0 in its default configuration. Encoding each fact in isolation
($w{=}0$) trails Mem0 and \emph{degrades} as more facts are retrieved, from
$62.2\%$ to $53.8\%$ on Llama and, most steeply, $54.0\%$ to $37.2\%$ on
Mistral as $k$ grows from $5$ to $50$. This is because independently encoded
facts share no cross-fact attention and add noise. A context window reverses
this: with $w{\ge}20$, accuracy is flat-to-rising in $k$ and comes within a few
points of Mem0 (on Llama, $60.3\%$ vs.\ $63.3\%$ at $k{=}50$) while often
exceeding it at small $k$. Prompt injection (Mem0), which performs full
cross-fact attention over the retrieved set, remains the accuracy ceiling, most
visibly on Mistral and Qwen. Because the window is a purely offline cost, this accuracy is
bought with no serving-time penalty.

The per-category breakdown (\Cref{app:accuracy}) localizes this residual gap.
It concentrates in \emph{multi-hop} questions, which require reasoning across
several retrieved facts, exactly the cross-fact attention that a joint prefill
provides and independent encoding omits (\Cref{thm:chunked-rope}(3)), where
Mem0 leads by as much as $16$ points at $k{=}50$ (steepest on Mistral). On
\emph{open-ended} questions the ordering reverses on Llama and Mistral, where
\sysname exceeds Mem0 by $3$--$9$ points; single-hop and temporal accuracy are
closer, with Mem0 generally holding a slight edge. The gap to prompt injection is therefore not a general
loss but a targeted one, confined to the multi-hop reasoning that most depends
on cross-fact attention.

\subsection{Serving throughput}
\label{sec:eval:throughput}

\begin{figure*}[t]
\centering
\begin{tikzpicture}
\begin{groupplot}[tputaxis, group style={group size=3 by 1, horizontal sep=1.0cm}]
\nextgroupplot[title={Llama-3.1-8B}, ylabel={Throughput (QPS)}]
\addplot[gputput] coordinates {(10,24.77)(25,46.61)(50,73.56)(75,86.35)(100,100.14)};
\addplot[cputput] coordinates {(10,24.78)(25,46.64)(50,73.47)(75,87.19)(100,101.91)};
\addplot[memtput] coordinates {(10,12.72)(25,19.20)(50,23.02)(75,24.98)(100,27.32)};
\nextgroupplot[title={Mistral-7B}]
\addplot[gputput] coordinates {(10,25.80)(25,47.19)(50,75.50)(75,90.80)(100,104.37)};
\addplot[cputput] coordinates {(10,25.81)(25,47.14)(50,75.04)(75,91.47)(100,105.07)};
\addplot[memtput] coordinates {(10,11.64)(25,16.70)(50,19.75)(75,21.37)(100,23.11)};
\nextgroupplot[title={Qwen2.5-7B}]
\addplot[gputput] coordinates {(10,29.93)(25,54.56)(50,87.54)(75,113.03)(100,134.56)};
\addplot[cputput] coordinates {(10,29.90)(25,54.64)(50,88.91)(75,113.51)(100,134.56)};
\addplot[memtput] coordinates {(10,14.71)(25,22.12)(50,26.68)(75,29.53)(100,32.44)};
\end{groupplot}
\end{tikzpicture}

\vspace{2pt}
\begin{tikzpicture}
\begin{axis}[hide axis, scale only axis, height=2mm, width=0.7\textwidth,
  xmin=0, xmax=1, ymin=0, ymax=1, legend columns=3,
  legend style={draw=none, at={(0.5,0.5)}, anchor=center, font=\footnotesize,
  /tikz/every even column/.append style={column sep=0.6em}}]
\addlegendimage{gputput}\addlegendentry{\sysname{} (GPU KV)}
\addlegendimage{cputput}\addlegendentry{\sysname{} (CPU-offloaded KV)}
\addlegendimage{memtput}\addlegendentry{Mem0}
\end{axis}
\end{tikzpicture}
\caption{Sustained serving throughput on \locomo (queries/s) vs.\ the number of
concurrent users, for \sysname (GPU-resident and CPU-offloaded KV) and Mem0.
\sysname's throughput scales near-linearly with concurrency, whereas Mem0
saturates; the GPU and CPU-offloaded curves overlap.}
\label{fig:throughput}
\end{figure*}
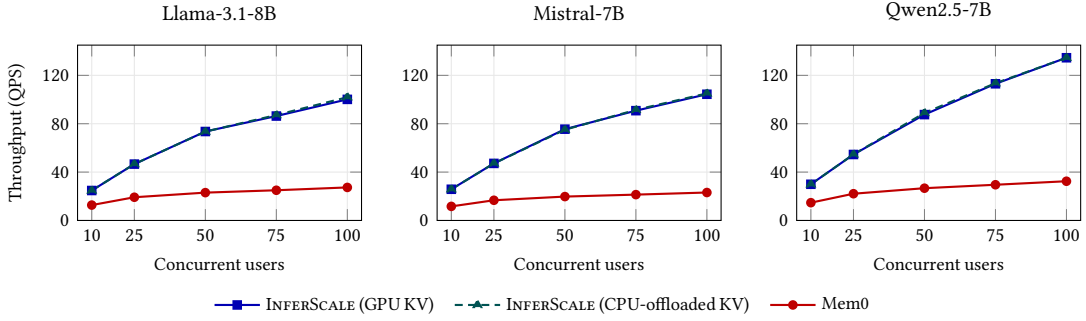

\para{Setup} 
We measure sustained serving throughput under concurrent load on \locomo
dataset. $N$ users share one conversation's memory corpus, each user issues two
distinct queries sampled from that conversation's question set, and all $2N$
requests are submitted to the vLLM engine as a single batch. Decoding is greedy
and every request generates exactly 50 output tokens, so conditions differ only
in how retrieved memory reaches the model. We sweep $N\in\{10,25,50,75,100\}$
and report QPS as completed requests divided by the engine's batch generation
time. 
% As in~\Cref{sec:eval:setup}, retrieval, KV composition, and prompt assembly
% are timed separately and excluded; 
The engine is warmed with length-matched random-token batches and the prefix
cache is reset before measurement, so no measured prompt is cache-seeded. 
Both systems retrieve top-$50$ facts per query: Mem0 re-prefills them as prompt
text, while \sysname injects the pre-composed KV.

\para{Results} \Cref{fig:throughput} shows that \sysname's throughput scales
near-linearly with the number of concurrent users, while Mem0 saturates early.
At $100$ users, \sysname reaches $100$/$104$/$135$ QPS on Llama/Mistral/Qwen
versus Mem0's $27$/$23$/$32$ QPS, a $3.7$--$4.5\times$ speedup, and the gap
widens with concurrency (Mem0 gains only ${\sim}2\times$ from $10$ to $100$
users, whereas \sysname gains ${\sim}4\times$). The reason is that \sysname
injects precomputed KV tensors and prefills only each request's short query, so many
more requests fit in a batch, Mem0 re-prefills the retrieved facts on every
request and bottlenecks on prefill compute. Offloading the KV store to CPU
costs no throughput loss: the CPU and GPU curves are indistinguishable, because the
one-time PCIe transfer is overlapped and amortized across the batch's decode.
% Thus the capacity headroom of CPU offloading (\Cref{sec:eval:cpu-offload}) comes
% for free in throughput.

% ------------------------------------------------------------
\subsection{Robustness to model scale}
\label{sec:eval:scale}

\begin{table}[t]
  \centering
  \small
  \caption{Robustness to model scale: \textbf{Qwen3-14B} on \locomo at $k{=}50$,
  $w{=}50$. Judged accuracy (\%) and latency (ms) for \sysname (GPU-resident and
  CPU-offloaded KV) vs.\ Mem0 in its default configuration (Qdrant vector store
  with prompt injection). Full sweeps
  over $k$ and $w$ and the per-category breakdown are in \Cref{app:qwen3-14b}.}
  \label{tab:qwen3-14b}
  \begin{tabular}{@{}lrrr@{}}
    \toprule
    Method & Accuracy (\%) & TTFT (ms) & Q-to-first-tok.\ (ms) \\
    \midrule
    \sysname (GPU KV) & 69.42 & 28.06 & 88.02 \\
    \sysname (CPU KV) & 69.42 & 31.11 & 129.55 \\
    Mem0              & 79.09 & 140.52 & 257.76 \\
    \bottomrule
  \end{tabular}
\end{table}

To test whether \sysname's techniques hold beyond the 7--8B regime, we repeat
the evaluation on \textbf{Qwen3-14B}. \Cref{tab:qwen3-14b} reports the main
results for ($k{=}50$, $w{=}50$). The full sweeps over $k$ and $w$ and the
per-category breakdown are in~\Cref{app:qwen3-14b}.

The results follow the same pattern as with other models evaluated. The latency
advantage, \emph{grows} with model size. KV injection holds TTFT flat at
${\sim}28$\,ms while Mem0 climbs to $141$\,ms ($5.0\times$ lower), and lowers
end-to-end query-to-first-token by $2.9\times$. This is due to the fact that a
larger model's prefill is costlier for prompt injection to repeat on every
request. CPU offloading remains inexpensive, adding only ${\sim}3$\,ms to TTFT. 
Regarding accuracy, context-window encoding is essential in recovering accuracy
loss due to chunked RoPE encoding (lack of cross attention in memory facts).
Encoding facts in isolation ($w{=}0$) collapses accuracy as $k$ grows (to
$24.7\%$ at $k{=}50$), while a window restores it to $69.4\%$, higher in
absolute terms than on the 7--8B parameter models. 
% though the gap to Mem0 ($76.6\%$) widens to ${\sim}7$ points, as the larger
% model exploits full cross-fact attention more effectively. 
%
\sysname's mechanism is therefore scale to larger models, and its serving benefit
strengthens at scale.

\para{Beyond conversational memory} \sysname's benefits are not specific to
conversational memory. We also evaluate two document-RAG benchmarks,
QASPER~\cite{qasper} and MultiHop-RAG~\cite{multihoprag}, in which the
retrieved unit is a $1024$-token document chunk rather than a short extracted
fact (\Cref{app:rag-benchmarks}). The gains there are substantially
\emph{stronger}: because each request conditions on ${\sim}15$K retrieved
tokens, KV injection cuts engine TTFT by $21.7$--$25.2\times$ relative to
prompt injection (versus $3.6$--$4.8\times$ on \locomo), while judged accuracy
matches or exceeds prompt injection on MultiHop-RAG for Llama-3.1-8B and
Mistral-7B.

% ------------------------------------------------------------
\subsection{Cost of context-window encoding}
\label{sec:eval:window-cost}

\begin{table}[t]
  \centering
  \footnotesize
  \setlength{\tabcolsep}{4pt}
  \caption{Average offline model KV-tensor precomputation time in seconds per conversation. Each value averages the complete $k \in \{5,10\}$ runs, with 10 LoCoMo conversations per run.}
  \label{tab:locomo-kv-precompute-time-by-window}
  \begin{tabular}{@{}lrrrr@{}}
    \toprule
    Model & $w=0$ & $w=5$ & $w=20$ & $w=50$ \\
    \midrule
    Llama 3.1 8B & 38.91 & 48.25 & 105.59 & 206.45 \\
    Mistral 7B v0.3 & 15.51 & 20.76 & 41.19 & 78.24 \\
    Qwen2.5 7B & 30.91 & 41.56 & 80.77 & 166.50 \\
    \bottomrule
  \end{tabular}
\end{table}

Context-window encoding improves accuracy (\Cref{sec:eval:accuracy}) at a
purely \emph{offline} cost. However, encoding each fact behind its $w$
preceding turns lengthens the one-time preprocessing forward pass.
\Cref{tab:locomo-kv-precompute-time-by-window} reports this cost per
conversation. It grows with the window, on Llama, from $39$\,s at $w{=}0$ to
$206$\,s at $w{=}50$ ($\sim$$5\times$), with the same trend on Mistral and
Qwen. This is because a larger prefix means more tokens per encoding pass. The
cost is paid once per conversation and amortized across all of that user's
requests. Crucially, $w$ does not affect the serving latency. It changes
neither the stored nor the injected token count, so TTFT is flat across $w$
(\Cref{fig:latency}) and the GPU footprint is unchanged
(\Cref{sec:eval:footprint}). The window is therefore a pure offline-compute
knob that trades one-time preprocessing time for accuracy, with no serving-time
or memory penalty.

% ------------------------------------------------------------
\subsection{GPU memory footprint}
\label{sec:eval:footprint}
\begin{table}[t]
  \centering
  \footnotesize
  \setlength{\tabcolsep}{4pt}
  \caption{Average per-conversation storage footprint in decimal MB. Jasper graph and fact-KV tensors are GPU metrics; the fact--ID map is a CPU metric.}
  \label{tab:locomo-kv-component-space}
  \begin{tabular}{@{}lrrr@{}}
    \toprule
    Model & Jasper GPU & Map CPU & Fact KVs GPU \\
    \midrule
    Llama 3.1 8B & 13.50 & 7.97 & 4796.67 \\
    Mistral 7B v0.3 & 13.50 & 3.02 & 2257.44 \\
    Qwen2.5 7B & 13.50 & 6.48 & 1780.10 \\
    \bottomrule
  \end{tabular}
\end{table}

\sysname keeps three structures resident per conversation
(\Cref{tab:locomo-kv-component-space}): the Jasper proximity
graph~\cite{mccoy2026gpu}, a CPU-side fact--ID map, and the pre-RoPE fact KV
store. The retrieval structures are negligible in size. The Jasper graph and
the map together cost under $~25$\,MB per conversation on every model. The
footprint is dominated by the KV store ($1.8$--$4.8$\,GB per conversation),
which is also the only term that scales with model architecture, growing with
the number of layers, KV heads, and head dimension: Qwen2.5-7B's aggressive
grouped-query attention makes it the cheapest ($1.78$\,GB) and Llama-3.1-8B the
most expensive ($4.80$\,GB). The KV store therefore bounds how many users can
stay resident at once. After model weights and vLLM's paged cache, a single
96\,GB GPU holds the memory of roughly ten (Llama) to several tens (Qwen) of
concurrent conversations, but because each fact is encoded once and reused
across all of a user's requests, this is a one-time, amortized cost rather than
a per-request one.

\subsection{Offloading KV embedding to CPU}
\label{sec:eval:cpu-offload}

\begin{table}[t]
  \centering
  \small
  \caption{Latency impact of offloading the pre-RoPE KV store to pinned host
  DRAM (streamed to the GPU over PCIe) vs.\ keeping it GPU-resident, at $w{=}5$
  and $k{=}50$ (the largest budget). All values in ms, averaged over 1{,}540
  queries. Offloading adds ${\sim}1.5$--$3$\,ms to engine TTFT but a larger,
  PCIe-bound ${\sim}33$--$38$\,ms to the end-to-end latencies at this budget.}
  \label{tab:cpu-offload}
  \begin{tabular}{@{}ll rrr@{}}
    \toprule
    Model & KV store & TTFT & Q-to-first-tok. & Q-to-full-answer \\
    \midrule
    \multirow{2}{*}{Llama-3.1-8B} & GPU & 17.52 & 62.12 & 274.02 \\
      & CPU & 20.45 & 100.29 & 312.66 \\
    \midrule
    \multirow{2}{*}{Mistral-7B} & GPU & 17.02 & 53.10 & 583.65 \\
      & CPU & 18.93 & 86.42 & 618.32 \\
    \midrule
    \multirow{2}{*}{Qwen2.5-7B} & GPU & 17.82 & 67.37 & 482.43 \\
      & CPU & 19.29 & 103.42 & 519.98 \\
    \bottomrule
  \end{tabular}
\end{table}

The per-conversation KV embedding is the biggest component that resides in GPU
memory and can become a bottleneck for supporting larger memory context. At
$1.8$--$4.8$\,GB it limits how many users' memory fits in GDDR. However, to
overcome this bottleneck we can  keep the pre-RoPE KV in pinned host DRAM and
stream only the retrieved facts' tensors to the GPU over PCIe on each request.
\Cref{tab:cpu-offload} measures what this costs at serving time.

Offloading carries a minor cost due to PCIe data transfer latency. Engine TTFT rises by only
${\sim}1.5$--$3$\,ms (e.g., $17.5\to20.5$\,ms on Llama at $k{=}50$), since the
injection step itself is unchanged. The end-to-end latencies, however, absorb
the host-to-device transfer of the retrieved facts' KV over PCIe:
query-to-first-token grows by ${\sim}33$--$38$\,ms at $k{=}50$ (e.g.,
$62\to100$\,ms on Llama), and query-to-full-answer by a similar amount. This
overhead scales with the retrieval budget, from ${\sim}8$\,ms at $k{=}5$ to
${\sim}38$\,ms at $k{=}50$ on Llama, because more retrieved KV must be streamed.

Even so, CPU-offloaded query-to-first-token ($\approx$$100$\,ms on Llama at
$k{=}50$) is $2.3\times$ faster than Mem0's ($236$\,ms). Offloading therefore
lifts the GDDR capacity bound at a per-request cost of tens of milliseconds that
scales with $k$ (full sweep in \Cref{app:cpu-offload}).

% ------------------------------------------------------------
\subsection{Empirical equivalence}
\label{sec:eval:equivalence}

\begin{table}[h]
\centering
\small
\caption{Empirical-equivalence summary: outcome counts over the six
(user $\times$ query-type) cases per model; factual accuracy is $100\%$ in all
cases. Full per-case breakdown in \Cref{tab:equivalence} of \Cref{app:equivalence}.}
\label{tab:equivalence-summary}
\begin{tabular}{@{}lccc@{}}
\toprule
Model & Identical & Sem.\ equiv. & Different \\
\midrule
Llama-3.1-8B & 3 & 3 & 0 \\
Mistral-7B   & 2 & 4 & 0 \\
Qwen2.5-7B   & 4 & 1 & 1 \\
\midrule
Total        & 9 & 8 & 1 \\
\bottomrule
\end{tabular}
\end{table}

\Cref{thm:equivalence} guarantees identical outputs under exact arithmetic. We
test whether this survives bf16 paged attention, where GPU reductions are
non-deterministic. On each model we build three users ($89$--$143$ memory
tokens) and issue a factual-recall and an open-ended query each ($18$ cases).
Both paths consume the \emph{identical} prompt tokens and differ only in whether
the memory prefix's KV is supplied by \sysname's connector or recomputed by
prefill, with greedy decoding. \Cref{tab:equivalence-summary} summarizes the
outcome: most cases are token-for-token identical or semantically equivalent,
one differs, and factual accuracy is $100\%$ throughout. This is
exactly what bf16 predicts: the two paths reach the same prefix KV through
different reduction orders, so a near-tie between logits can flip the greedy
argmax and cascade without changing factual grounding, confirming
\Cref{thm:equivalence} in practice.

% ============================================================
% Intentionally omitted per the KV-vs-PI scope (previously stale):
%   - full-context baseline, standalone Mem0 comparison (PI already == Mem0),
%   - multi-user throughput, CPU-offload / scale-out.
% Re-introduce only if re-run under the new framing.
% ============================================================

% \input{related}   % folded into the last paragraph of the intro (overlap with intro/§2)
% ============================================================
\section{Conclusion}
\label{sec:conclusion}
% ============================================================

We presented \sysname, a GPU-native memory system that replaces the repeated
prefill of retrieved memory with reusable KV state: it precomputes each fact's
KV once and injects it into vLLM's paged cache, made
% exact and 
relocatable by chunked RoPE and made accurate by context-window encoding, with
no fine-tuning, and vLLM engine changes.
% or retrieval-pipeline changes.
We observe that \sysname can extend to any workload that repeatedly conditions
on a largely static retrieved context such as RAG over a stable corpus, cached
system prompts, tool descriptions, and we expect attention-layer injection to
become a default primitive for conditioning LLMs on static context.

Across three open-weight models on \locomo, \sysname keeps engine TTFT nearly
constant as the retrieval budget grows---$3.6$--$4.8\times$ lower than Mem0 at
$k{=}50$---and raises throughput by $3.7$--$4.5\times$ at $100$ concurrent
users, while context-window encoding recovers accuracy to within a few points of
prompt injection ($60.3\%$ vs.\ $63.3\%$). Offloading the KV store to host DRAM
further lifts the GPU-memory capacity bound at only a few-millisecond TTFT cost.

% \para{Future work}
\sysname opens several directions. First, its GPU-resident KV store is the
capacity bottleneck; scaling to larger contexts and more users calls for
\emph{eviction and paging} policies that keep hot memory on the GPU, spill cold
memory to host DRAM or disk, and prefetch along the retrieval path. Second, our
current store is built once and treated as read-only; supporting \emph{online
memory updates}, incrementally inserting, revising, and deleting facts, with
consistent updates to both the Jasper index and the pre-RoPE KV store, would
let \sysname track memory that evolves as conversations continue. Third,
extending \sysname to \emph{multi-node, multi-GPU} deployments raises questions
of sharding the KV store and retrieval index across devices and routing requests
to where a user's memory resides, so that capacity and throughput scale beyond a
single accelerator.

% We presented a GPU-native memory system that injects pre-computed KV cache
% tensors directly into vLLM's paged attention system, eliminating the $O(m^2)$
% prefill overhead of prompt injection. Our system achieves 4.39$\times$ TTFT
% improvement at 32K memory tokens and up to 7.59$\times$ throughput improvement
% under concurrent multi-user load, with semantically equivalent output quality.
% Cross-model validation on Mistral-7B-Instruct-v0.3 and Qwen2.5-7B-Instruct
% confirms these gains are architecture-agnostic, and evaluation on the LOCOMO
% benchmark demonstrates that KV injection matches full-context QA quality while
% reducing latency by 5.56$\times$.
%
% The approach is general: any scenario where static context is repeatedly
% prepended to requests---user memories, system prompts, few-shot examples,
% tool descriptions---can benefit from KV injection. We believe attention-layer
% injection will become the standard approach for personalized LLM serving as
% per-user context sizes continue to grow.

\begin{acks}
This research is funded in part by NSF grant OAC 2339521 and 2517201.
\end{acks}

% ============================================================
% References
% ============================================================
\bibliographystyle{ACM-Reference-Format}
\balance % even out the two columns of the bibliography
\bibliography{references}

@inproceedings{kwon2023pagedattention,
  author    = {Woosuk Kwon and Zhuohan Li and Siyuan Zhuang and Ying Sheng and Lianmin Zheng and Cody Hao Yu and Joseph E. Gonzalez and Hao Zhang and Ion Stoica},
  title     = {Efficient Memory Management for Large Language Model Serving with {PagedAttention}},
  booktitle = {Proceedings of the 29th Symposium on Operating Systems Principles},
  series    = {SOSP '23},
  year      = {2023},
  pages     = {611--626},
  publisher = {Association for Computing Machinery},
  address   = {New York, NY, USA},
  location  = {Koblenz, Germany},
  doi       = {10.1145/3600006.3613165},
  url       = {https://doi.org/10.1145/3600006.3613165}
}

@misc{letta,
  author        = {Charles Packer and Sarah Wooders and Kevin Lin and Vivian Fang and Shishir G. Patil and Ion Stoica and Joseph E. Gonzalez},
  title         = {{MemGPT}: Towards {LLM}s as Operating Systems},
  year          = {2023},
  eprint        = {2310.08560},
  archiveprefix = {arXiv},
  doi           = {10.48550/arXiv.2310.08560},
  url           = {https://arxiv.org/abs/2310.08560}
}

@article{su2021roformer,
  author  = {Su, Jianlin and Ahmed, Murtadha and Lu, Yu and Pan, Shengfeng and Bo, Wen and Liu, Yunfeng},
  title   = {{RoFormer}: Enhanced Transformer with Rotary Position Embedding},
  journal = {Neurocomputing},
  volume  = {568},
  pages   = {127063},
  year    = {2024},
  doi     = {10.1016/j.neucom.2023.127063},
  url     = {https://doi.org/10.1016/j.neucom.2023.127063}
}

@inproceedings{yao2024cacheblend,
  author    = {Jiayi Yao and Hanchen Li and Yuhan Liu and Siddhant Ray and Yihua Cheng and Qizheng Zhang and Kuntai Du and Shan Lu and Junchen Jiang},
  title     = {{CacheBlend}: Fast Large Language Model Serving for {RAG} with Cached Knowledge Fusion},
  booktitle = {Proceedings of the Twentieth European Conference on Computer Systems},
  series    = {EuroSys '25},
  year      = {2025},
  pages     = {94--109},
  publisher = {Association for Computing Machinery},
  address   = {New York, NY, USA},
  location  = {Rotterdam, Netherlands},
  doi       = {10.1145/3689031.3696098},
  url       = {https://doi.org/10.1145/3689031.3696098}
}

@inproceedings{locomo,
  author    = {Adyasha Maharana and Dong-Ho Lee and Sergey Tulyakov and Mohit Bansal and Francesco Barbieri and Yuwei Fang},
  title     = {Evaluating Very Long-Term Conversational Memory of {LLM} Agents},
  booktitle = {Proceedings of the 62nd Annual Meeting of the Association for Computational Linguistics (Volume 1: Long Papers)},
  year      = {2024},
  month     = aug,
  pages     = {13851--13870},
  publisher = {Association for Computational Linguistics},
  address   = {Bangkok, Thailand},
  doi       = {10.18653/v1/2024.acl-long.747},
  url       = {https://aclanthology.org/2024.acl-long.747/}
}

@inproceedings{qasper,
  author    = {Pradeep Dasigi and Kyle Lo and Iz Beltagy and Arman Cohan and Noah A. Smith and Matt Gardner},
  title     = {A Dataset of Information-Seeking Questions and Answers Anchored in Research Papers},
  booktitle = {Proceedings of the 2021 Conference of the North American Chapter of the Association for Computational Linguistics: Human Language Technologies},
  year      = {2021},
  month     = jun,
  pages     = {4599--4610},
  publisher = {Association for Computational Linguistics},
  doi       = {10.18653/v1/2021.naacl-main.365},
  url       = {https://aclanthology.org/2021.naacl-main.365/}
}

@inproceedings{multihoprag,
  author    = {Yixuan Tang and Yi Yang},
  title     = {{MultiHop-RAG}: Benchmarking Retrieval-Augmented Generation for Multi-Hop Queries},
  booktitle = {Proceedings of the First Conference on Language Modeling (COLM)},
  year      = {2024},
  eprint    = {2401.15391},
  archiveprefix = {arXiv},
  url       = {https://arxiv.org/abs/2401.15391}
}

@misc{gemma2024,
  author        = {{Gemma Team}},
  title         = {{Gemma 2}: Improving Open Language Models at a Practical Size},
  year          = {2024},
  eprint        = {2408.00118},
  archiveprefix = {arXiv},
  doi           = {10.48550/arXiv.2408.00118},
  url           = {https://arxiv.org/abs/2408.00118}
}

@misc{chhikara2025mem0,
  author        = {Prateek Chhikara and Dev Khant and Saket Aryan and Taranjeet Singh and Deshraj Yadav},
  title         = {{Mem0}: Building Production-Ready {AI} Agents with Scalable Long-Term Memory},
  year          = {2025},
  eprint        = {2504.19413},
  archiveprefix = {arXiv},
  doi           = {10.48550/arXiv.2504.19413},
  url           = {https://arxiv.org/abs/2504.19413}
}

@misc{mccoy2026gpu,
  author        = {Hunter McCoy and Zikun Wang and Prashant Pandey},
  title         = {{GPU}-Accelerated {ANNS}: Quantized for Speed, Built for Change},
  year          = {2026},
  eprint        = {2601.07048},
  archiveprefix = {arXiv},
  doi           = {10.48550/arXiv.2601.07048},
  url           = {https://arxiv.org/abs/2601.07048}
}

@misc{zep2025,
  author        = {Preston Rasmussen and Pavlo Paliychuk and Travis Beauvais and Jack Ryan and Daniel Chalef},
  title         = {{Zep}: A Temporal Knowledge Graph Architecture for Agent Memory},
  year          = {2025},
  eprint        = {2501.13956},
  archiveprefix = {arXiv},
  doi           = {10.48550/arXiv.2501.13956},
  url           = {https://arxiv.org/abs/2501.13956}
}

@article{retroinfer2025,
  author  = {Yaoqi Chen and Jinkai Zhang and Baotong Lu and Qianxi Zhang and Chengruidong Zhang and Jing Liu and Jingjia Luo and Di Liu and Huiqiang Jiang and Qi Chen and Bailu Ding and Xiao Yan and Jiawei Jiang and Chen Chen and Mingxing Zhang and Cheng Li and Yuqing Yang and Fan Yang and Mao Yang},
  title   = {{RetroInfer}: A Vector Storage Engine for Scalable Long-Context {LLM} Inference},
  journal = {Proceedings of the VLDB Endowment},
  volume  = {19},
  number  = {5},
  pages   = {1016--1031},
  year    = {2026},
  doi     = {10.14778/3796195.3796212},
  url     = {https://www.vldb.org/pvldb/vol19/p1016-lu.pdf}
}

@inproceedings{retrievalattention2024,
  author    = {Di Liu and Meng Chen and Baotong Lu and Huiqiang Jiang and Zhenhua Han and Qianxi Zhang and Qi Chen and Chengruidong Zhang and Bailu Ding and Kai Zhang and Chen Chen and Fan Yang and Yuqing Yang and Lili Qiu},
  title     = {{RetrievalAttention}: Accelerating Long-Context {LLM} Inference via Vector Retrieval},
  booktitle = {Advances in Neural Information Processing Systems},
  volume    = {38},
  year      = {2025},
  numpages  = {28},
  publisher = {Curran Associates, Inc.},
  address   = {Red Hook, NY, USA},
  url       = {https://proceedings.neurips.cc/paper_files/paper/2025/hash/4e36d4049fb0fea195a8267c8dcd0824-Abstract-Conference.html}
}

@inproceedings{blockattention,
  author    = {Dongyang Ma and Yan Wang and Tian Lan},
  title     = {{Block-Attention} for Efficient Prefilling},
  booktitle = {The Thirteenth International Conference on Learning Representations},
  year      = {2025},
  numpages  = {15},
  publisher = {OpenReview.net},
  address   = {Singapore},
  url       = {https://proceedings.iclr.cc/paper_files/paper/2025/hash/a03037317560b8c5f2fb4b6466d4c439-Abstract-Conference.html}
}

@misc{lazyattention,
  author        = {Haocheng Xia and Mihir Pamnani and Hanxi Fang and Supawit Chockchowwat and Yongjoo Park},
  title         = {{LazyAttention}: Efficient Retrieval-Augmented Generation with Deferred Positional Encoding},
  year          = {2026},
  eprint        = {2606.04302},
  archiveprefix = {arXiv},
  doi           = {10.48550/arXiv.2606.04302},
  url           = {https://arxiv.org/abs/2606.04302}
}

@misc{lmcache,
  author        = {Yuhan Liu and Yihua Cheng and Jiayi Yao and Yuwei An and Xiaokun Chen and Shaoting Feng and Yuyang Huang and Samuel Shen and Rui Zhang and Kuntai Du and Junchen Jiang},
  title         = {{LMCache}: An Efficient {KV} Cache Layer for Enterprise-Scale {LLM} Inference},
  year          = {2025},
  eprint        = {2510.09665},
  archiveprefix = {arXiv},
  doi           = {10.48550/arXiv.2510.09665},
  url           = {https://arxiv.org/abs/2510.09665}
}

@misc{qdrant,
  author       = {{Qdrant Solutions GmbH}},
  title        = {{Qdrant}: Vector Database and Vector Search Engine},
  year         = {2026},
  howpublished = {Software repository},
  url          = {https://github.com/qdrant/qdrant},
  note         = {Accessed 17 July 2026}
}

@misc{llama3,
  author        = {Llama Team and AI @ Meta},
  title         = {The {Llama 3} Herd of Models},
  year          = {2024},
  eprint        = {2407.21783},
  archiveprefix = {arXiv},
  doi           = {10.48550/arXiv.2407.21783},
  url           = {https://arxiv.org/abs/2407.21783}
}

@misc{mistral7b,
  author       = {{Mistral AI Team}},
  title        = {{Mistral-7B-Instruct-v0.3} Model Card},
  year         = {2024},
  howpublished = {Hugging Face model repository},
  url          = {https://huggingface.co/mistralai/Mistral-7B-Instruct-v0.3},
  note         = {Accessed 17 July 2026}
}

@misc{qwen25,
  author        = {{Qwen Team}},
  title         = {{Qwen2.5} Technical Report},
  year          = {2024},
  eprint        = {2412.15115},
  archiveprefix = {arXiv},
  doi           = {10.48550/arXiv.2412.15115},
  url           = {https://arxiv.org/abs/2412.15115}
}

@misc{qwen3,
  author        = {Qwen Team},
  title         = {{Qwen3} Technical Report},
  year          = {2025},
  eprint        = {2505.09388},
  archiveprefix = {arXiv},
  doi           = {10.48550/arXiv.2505.09388},
  url           = {https://arxiv.org/abs/2505.09388}
}

@misc{neelakantan2022embeddings,
  author       = {{OpenAI}},
  title        = {New Embedding Models and {API} Updates},
  year         = {2024},
  month        = jan,
  howpublished = {OpenAI product announcement},
  url          = {https://openai.com/index/new-embedding-models-and-api-updates/},
  note         = {Published 25 January 2024; accessed 17 July 2026}
}

@inproceedings{ainslie2023gqa,
  title     = {{GQA}: Training Generalized Multi-Query Transformer Models from Multi-Head Checkpoints},
  author    = {Ainslie, Joshua and Lee-Thorp, James and de Jong, Michiel and Zemlyanskiy, Yury and Lebr{\'o}n, Federico and Sanghai, Sumit},
  booktitle = {Proceedings of the 2023 Conference on Empirical Methods in Natural Language Processing (EMNLP)},
  year      = {2023}
}

@article{malkov2020hnsw,
  title   = {Efficient and Robust Approximate Nearest Neighbor Search Using Hierarchical Navigable Small World Graphs},
  author  = {Malkov, Yu A. and Yashunin, Dmitry A.},
  journal = {IEEE Transactions on Pattern Analysis and Machine Intelligence},
  volume  = {42},
  number  = {4},
  pages   = {824--836},
  year    = {2020}
}

@article{johnson2019faiss,
  title   = {Billion-Scale Similarity Search with {GPUs}},
  author  = {Johnson, Jeff and Douze, Matthijs and J{\'e}gou, Herv{\'e}},
  journal = {IEEE Transactions on Big Data},
  year    = {2019}
}

@article{ootomo2023cagra,
  title   = {{CAGRA}: Highly Parallel Graph Construction and Approximate Nearest Neighbor Search for {GPUs}},
  author  = {Ootomo, Hiroyuki and Naruse, Akira and Nolet, Corey and Wang, Ray and Feher, Tamas and Wang, Yong},
  journal = {arXiv preprint arXiv:2308.15136},
  year    = {2023}
}

% ============================================================
% Appendix
% ============================================================
\clearpage
\appendix
% Turn balancing back off: the \balance before the bibliography stays active
% until here; without \nobalance the table-heavy appendix's last page gets
% force-balanced and overflows the bottom margin.
\nobalance
\crefalias{section}{appendix} % \Cref to appendix sections -> "Appendix A/B/..."
% \emergencystretch absorbs the one long text line in the appendix (a \Cref that
% can't be hyphenated) so it doesn't overflow the column edge.
\emergencystretch=1em
\begin{center}
  {\Large\bfseries Appendix}
\end{center}
\section{Full Serving-Latency Results}
\label{app:latency}

\Cref{fig:latency} plots serving latency for a representative encoding window;
\Cref{tab:latency-full} gives the full per-window breakdown---engine TTFT and
end-to-end query-to-first-token for \sysname at every window $w$ and for Mem0,
across all three models. \sysname's latency is essentially unchanged across $w$
(within $1$\,ms), confirming that the encoding window is an offline-only cost.

\begin{table*}[t]
  \centering
  \small
  \caption{Full serving latency on \locomo (ms, averaged over 1{,}540 queries):
  engine TTFT and end-to-end query-to-first-token for \sysname (per encoding
  window $w$) and Mem0, with prefix caching enabled.
  \Cref{fig:latency} plots a representative window ($w{=}5$).}
  \label{tab:latency-full}
  \begin{tabular}{@{}ll rrrr rrrr@{}}
    \toprule
    & & \multicolumn{4}{c}{TTFT (ms)} & \multicolumn{4}{c}{Query-to-first-token (ms)} \\
    \cmidrule(lr){3-6}\cmidrule(lr){7-10}
    Model & Method & $k{=}5$ & $10$ & $20$ & $50$ & $k{=}5$ & $10$ & $20$ & $50$ \\
    \midrule
    \multirow{5}{*}{Llama-3.1-8B}
      & \sysname ($w{=}0$)  & 16.57 & 16.48 & 16.66 & 17.33 & 52.30 & 51.92 & 53.14 & 60.34 \\
      & \sysname ($w{=}5$)  & 16.49 & 16.50 & 16.87 & 17.52 & 51.94 & 51.72 & 55.52 & 62.12 \\
      & \sysname ($w{=}20$) & 16.19 & 16.46 & 17.06 & 17.28 & 50.71 & 51.50 & 56.29 & 59.78 \\
      & \sysname ($w{=}50$) & 16.33 & 16.42 & 16.69 & 17.26 & 51.19 & 51.74 & 53.82 & 59.76 \\
      & Mem0                & 33.15 & 39.03 & 46.44 & 68.33 & 186.15 & 190.22 & 206.79 & 236.48 \\
    \midrule
    \multirow{5}{*}{Mistral-7B}
      & \sysname ($w{=}0$)  & 15.73 & 15.80 & 15.96 & 17.07 & 44.83 & 44.72 & 45.85 & 54.77 \\
      & \sysname ($w{=}5$)  & 15.42 & 15.68 & 16.07 & 16.81 & 42.66 & 43.77 & 46.58 & 53.18 \\
      & \sysname ($w{=}20$) & 15.65 & 15.72 & 16.00 & 16.66 & 43.55 & 43.51 & 46.07 & 52.21 \\
      & \sysname ($w{=}50$) & 15.58 & 15.64 & 16.12 & 16.82 & 43.17 & 43.63 & 46.87 & 52.64 \\
      & Mem0                & 39.15 & 53.88 & 63.08 & 81.71 & 129.37 & 147.41 & 162.01 & 198.07 \\
    \midrule
    \multirow{5}{*}{Qwen2.5-7B}
      & \sysname ($w{=}0$)  & 16.61 & 16.73 & 17.22 & 17.61 & 55.81 & 56.92 & 59.84 & 65.69 \\
      & \sysname ($w{=}5$)  & 16.90 & 16.71 & 17.09 & 17.82 & 58.46 & 56.66 & 60.16 & 67.37 \\
      & \sysname ($w{=}20$) & 16.80 & 16.98 & 17.22 & 17.57 & 56.97 & 58.30 & 59.99 & 64.59 \\
      & \sysname ($w{=}50$) & 16.77 & 17.02 & 17.10 & 17.95 & 56.76 & 58.55 & 59.16 & 66.92 \\
      & Mem0                & 34.64 & 36.15 & 43.87 & 62.89 & 178.81 & 182.25 & 190.83 & 224.76 \\
    \bottomrule
  \end{tabular}
\end{table*}

\section{Per-Category LoCoMo Accuracy}
\label{app:accuracy}

\Cref{tab:accuracy} reports overall judged accuracy;
\Cref{tab:accuracy-percat-llama,tab:accuracy-percat-mistral,tab:accuracy-percat-qwen}
break it down by \locomo question category for all three models, comparing
\sysname (per encoding window $w$) against Mem0. The
per-category trends match the overall result: the $w{=}0$ configuration degrades
as $k$ grows, a larger window recovers most of the loss, and Mem0 is the ceiling
at large $k$.

\begin{table}[H]
  \centering
  \small
  \caption{Per-category judged \locomo accuracy (\%) vs.\ retrieval budget $k$
  for \textbf{Llama-3.1-8B}: \sysname (per encoding window $w$) vs.\ Mem0
  (prompt injection).}
  \label{tab:accuracy-percat-llama}
  \begin{tabular}{@{}ll rrrr@{}}
    \toprule
    Category & Method & $k{=}5$ & $10$ & $20$ & $50$ \\
    \midrule
    \multirow{5}{*}{Single-hop}
      & \sysname ($w{=}0$)  & 72.29 & 70.75 & 67.78 & 64.21 \\
      & \sysname ($w{=}5$)  & 70.51 & 70.63 & 71.46 & 67.42 \\
      & \sysname ($w{=}20$) & 72.18 & 73.84 & 73.37 & 71.22 \\
      & \sysname ($w{=}50$) & 72.06 & 72.89 & 75.27 & 72.41 \\
      & Mem0                & 68.97 & 72.06 & 75.51 & 78.60 \\
    \midrule
    \multirow{5}{*}{Multi-hop}
      & \sysname ($w{=}0$)  & 56.74 & 59.57 & 52.84 & 55.67 \\
      & \sysname ($w{=}5$)  & 54.96 & 52.13 & 57.45 & 54.96 \\
      & \sysname ($w{=}20$) & 54.96 & 56.74 & 60.64 & 54.61 \\
      & \sysname ($w{=}50$) & 57.45 & 57.80 & 60.28 & 57.80 \\
      & Mem0                & 51.42 & 56.74 & 57.45 & 58.16 \\
    \midrule
    \multirow{5}{*}{Open-ended}
      & \sysname ($w{=}0$)  & 63.54 & 60.42 & 57.29 & 48.96 \\
      & \sysname ($w{=}5$)  & 62.50 & 61.46 & 61.46 & 57.29 \\
      & \sysname ($w{=}20$) & 58.33 & 63.54 & 62.50 & 63.54 \\
      & \sysname ($w{=}50$) & 54.17 & 63.54 & 61.46 & 60.42 \\
      & Mem0                & 52.08 & 51.04 & 54.17 & 54.17 \\
    \midrule
    \multirow{5}{*}{Temporal}
      & \sysname ($w{=}0$)  & 40.19 & 42.37 & 35.51 & 26.17 \\
      & \sysname ($w{=}5$)  & 36.14 & 37.69 & 31.15 & 30.53 \\
      & \sysname ($w{=}20$) & 33.64 & 35.51 & 31.46 & 30.53 \\
      & \sysname ($w{=}50$) & 32.71 & 33.64 & 32.09 & 30.53 \\
      & Mem0                & 31.78 & 30.53 & 30.53 & 30.22 \\
    \bottomrule
  \end{tabular}
\end{table}

\begin{table}[H]
  \centering
  \small
  \caption{Per-category judged \locomo accuracy (\%) vs.\ retrieval budget $k$
  for \textbf{Mistral-7B}: \sysname (per encoding window $w$) vs.\ Mem0
  (prompt injection).}
  \label{tab:accuracy-percat-mistral}
  \begin{tabular}{@{}ll rrrr@{}}
    \toprule
    Category & Method & $k{=}5$ & $10$ & $20$ & $50$ \\
    \midrule
    \multirow{5}{*}{Single-hop}
      & \sysname ($w{=}0$)  & 65.24 & 62.54 & 60.22 & 44.46 \\
      & \sysname ($w{=}5$)  & 64.45 & 67.54 & 67.38 & 68.93 \\
      & \sysname ($w{=}20$) & 62.90 & 66.71 & 66.59 & 70.04 \\
      & \sysname ($w{=}50$) & 65.16 & 66.47 & 68.13 & 70.51 \\
      & Mem0                & 72.41 & 74.08 & 75.74 & 74.20 \\
    \midrule
    \multirow{5}{*}{Multi-hop}
      & \sysname ($w{=}0$)  & 42.35 & 57.80 & 52.67 & 36.88 \\
      & \sysname ($w{=}5$)  & 50.00 & 51.42 & 54.96 & 49.46 \\
      & \sysname ($w{=}20$) & 47.52 & 47.16 & 52.48 & 51.61 \\
      & \sysname ($w{=}50$) & 51.06 & 48.94 & 53.90 & 52.31 \\
      & Mem0                & 56.03 & 64.54 & 65.60 & 68.09 \\
    \midrule
    \multirow{5}{*}{Open-ended}
      & \sysname ($w{=}0$)  & 56.25 & 54.17 & 56.25 & 41.05 \\
      & \sysname ($w{=}5$)  & 54.17 & 58.33 & 54.17 & 54.17 \\
      & \sysname ($w{=}20$) & 57.29 & 58.33 & 53.12 & 60.42 \\
      & \sysname ($w{=}50$) & 55.21 & 57.29 & 57.29 & 60.42 \\
      & Mem0                & 57.29 & 57.29 & 55.21 & 57.29 \\
    \midrule
    \multirow{5}{*}{Temporal}
      & \sysname ($w{=}0$)  & 34.27 & 33.64 & 32.60 & 17.45 \\
      & \sysname ($w{=}5$)  & 39.88 & 36.76 & 37.69 & 30.53 \\
      & \sysname ($w{=}20$) & 35.20 & 34.89 & 36.76 & 32.29 \\
      & \sysname ($w{=}50$) & 35.51 & 36.76 & 38.63 & 30.72 \\
      & Mem0                & 38.01 & 38.94 & 39.25 & 37.38 \\
    \bottomrule
  \end{tabular}
\end{table}

\begin{table}[H]
  \centering
  \small
  \caption{Per-category judged \locomo accuracy (\%) vs.\ retrieval budget $k$
  for \textbf{Qwen2.5-7B}: \sysname (per encoding window $w$) vs.\ Mem0
  (prompt injection).}
  \label{tab:accuracy-percat-qwen}
  \begin{tabular}{@{}ll rrrr@{}}
    \toprule
    Category & Method & $k{=}5$ & $10$ & $20$ & $50$ \\
    \midrule
    \multirow{5}{*}{Single-hop}
      & \sysname ($w{=}0$)  & 70.63 & 70.63 & 67.30 & 61.83 \\
      & \sysname ($w{=}5$)  & 70.87 & 70.27 & 70.63 & 70.51 \\
      & \sysname ($w{=}20$) & 71.94 & 73.13 & 72.06 & 73.84 \\
      & \sysname ($w{=}50$) & 72.65 & 71.46 & 71.58 & 73.72 \\
      & Mem0                & 73.72 & 77.53 & 79.79 & 80.14 \\
    \midrule
    \multirow{5}{*}{Multi-hop}
      & \sysname ($w{=}0$)  & 49.65 & 45.39 & 43.62 & 43.26 \\
      & \sysname ($w{=}5$)  & 53.90 & 48.58 & 53.55 & 52.13 \\
      & \sysname ($w{=}20$) & 52.84 & 52.48 & 51.06 & 51.06 \\
      & \sysname ($w{=}50$) & 48.94 & 51.42 & 52.48 & 51.42 \\
      & Mem0                & 53.19 & 57.80 & 60.28 & 59.93 \\
    \midrule
    \multirow{5}{*}{Open-ended}
      & \sysname ($w{=}0$)  & 64.58 & 58.33 & 58.33 & 56.25 \\
      & \sysname ($w{=}5$)  & 59.38 & 59.38 & 62.50 & 63.54 \\
      & \sysname ($w{=}20$) & 60.42 & 57.29 & 57.29 & 57.29 \\
      & \sysname ($w{=}50$) & 60.42 & 59.38 & 58.33 & 58.33 \\
      & Mem0                & 59.38 & 62.50 & 67.71 & 62.50 \\
    \midrule
    \multirow{5}{*}{Temporal}
      & \sysname ($w{=}0$)  & 38.63 & 29.91 & 28.04 & 25.23 \\
      & \sysname ($w{=}5$)  & 33.64 & 31.46 & 27.73 & 25.86 \\
      & \sysname ($w{=}20$) & 34.89 & 31.46 & 28.66 & 25.86 \\
      & \sysname ($w{=}50$) & 30.84 & 29.28 & 28.04 & 26.17 \\
      & Mem0                & 30.53 & 32.71 & 30.22 & 28.66 \\
    \bottomrule
  \end{tabular}
\end{table}

\section{Full CPU-Offloading Latency}
\label{app:cpu-offload}

\Cref{tab:cpu-offload} summarizes the effect of holding the pre-RoPE KV store in
host DRAM at a representative operating point;
\Cref{tab:cpu-offload-ttft,tab:cpu-offload-qtft,tab:cpu-offload-qta} give the
full sweep across the retrieval budget $k$ and the encoding window $w$, for
GPU-resident (\sysname's default) and CPU-resident KV, on all three models and
all three latency metrics. The two tracks are close on engine TTFT (within a few
ms), but CPU offload adds a larger, $k$-dependent overhead to the end-to-end
latencies---tens of ms at $k{=}50$---as the retrieved KV is streamed over PCIe.

\begin{table}[H]
  \centering
  \footnotesize
  \setlength{\tabcolsep}{4pt}
  \caption{Engine TTFT (ms), GPU-resident vs.\ CPU-resident KV store, per
  encoding window $w$ and retrieval budget $k$.}
  \label{tab:cpu-offload-ttft}
  \begin{tabular}{@{}ll rrrr@{}}
    \toprule
    Model & KV store & $k{=}5$ & $10$ & $20$ & $50$ \\
    \midrule
    \multirow{8}{*}{Llama-3.1-8B}
      & GPU ($w{=}0$)  & 16.57 & 16.48 & 16.66 & 17.33 \\
      & GPU ($w{=}5$)  & 16.49 & 16.50 & 16.87 & 17.52 \\
      & GPU ($w{=}20$) & 16.19 & 16.46 & 17.06 & 17.28 \\
      & GPU ($w{=}50$) & 16.33 & 16.42 & 16.69 & 17.26 \\
      & CPU ($w{=}0$)  & 16.49 & 16.97 & 18.26 & 20.72 \\
      & CPU ($w{=}5$)  & 17.46 & 17.15 & 18.06 & 20.45 \\
      & CPU ($w{=}20$) & 17.03 & 17.03 & 17.83 & 20.40 \\
      & CPU ($w{=}50$) & 16.86 & 17.35 & 18.22 & 20.48 \\
    \midrule
    \multirow{8}{*}{Mistral-7B}
      & GPU ($w{=}0$)  & 15.73 & 16.03 & 16.32 & 17.03 \\
      & GPU ($w{=}5$)  & 15.96 & 16.32 & 16.17 & 17.02 \\
      & GPU ($w{=}20$) & 15.83 & 15.94 & 16.23 & 16.92 \\
      & GPU ($w{=}50$) & 15.81 & 16.34 & 16.08 & 17.20 \\
      & CPU ($w{=}0$)  & 16.02 & 16.39 & 16.97 & 19.00 \\
      & CPU ($w{=}5$)  & 15.98 & 16.25 & 16.64 & 18.93 \\
      & CPU ($w{=}20$) & 16.11 & 16.43 & 16.53 & 18.88 \\
      & CPU ($w{=}50$) & 16.25 & 16.19 & 16.50 & 19.04 \\
    \midrule
    \multirow{8}{*}{Qwen2.5-7B}
      & GPU ($w{=}0$)  & 16.61 & 16.73 & 17.22 & 17.61 \\
      & GPU ($w{=}5$)  & 16.90 & 16.71 & 17.09 & 17.82 \\
      & GPU ($w{=}20$) & 16.80 & 16.98 & 17.22 & 17.57 \\
      & GPU ($w{=}50$) & 16.77 & 17.02 & 17.10 & 17.95 \\
      & CPU ($w{=}0$)  & 16.51 & 16.95 & 17.70 & 19.81 \\
      & CPU ($w{=}5$)  & 16.50 & 17.25 & 17.90 & 19.29 \\
      & CPU ($w{=}20$) & 16.35 & 17.14 & 18.00 & 19.78 \\
      & CPU ($w{=}50$) & 16.51 & 17.08 & 17.79 & 19.47 \\
    \bottomrule
  \end{tabular}
\end{table}

\begin{table}[H]
  \centering
  \footnotesize
  \setlength{\tabcolsep}{4pt}
  \caption{End-to-end query-to-first-token (ms), GPU-resident vs.\ CPU-resident
  KV store, per encoding window $w$ and retrieval budget $k$.}
  \label{tab:cpu-offload-qtft}
  \begin{tabular}{@{}ll rrrr@{}}
    \toprule
    Model & KV store & $k{=}5$ & $10$ & $20$ & $50$ \\
    \midrule
    \multirow{8}{*}{Llama-3.1-8B}
      & GPU ($w{=}0$)  & 52.30 & 51.92 & 53.14 & 60.34 \\
      & GPU ($w{=}5$)  & 51.94 & 51.72 & 55.52 & 62.12 \\
      & GPU ($w{=}20$) & 50.71 & 51.50 & 56.29 & 59.78 \\
      & GPU ($w{=}50$) & 51.19 & 51.74 & 53.82 & 59.76 \\
      & CPU ($w{=}0$)  & 54.49 & 58.15 & 70.28 & 101.70 \\
      & CPU ($w{=}5$)  & 60.16 & 59.81 & 70.58 & 100.29 \\
      & CPU ($w{=}20$) & 56.94 & 59.53 & 68.13 & 97.52 \\
      & CPU ($w{=}50$) & 56.42 & 61.38 & 70.50 & 99.00 \\
    \midrule
    \multirow{8}{*}{Mistral-7B}
      & GPU ($w{=}0$)  & 44.23 & 45.29 & 47.41 & 53.49 \\
      & GPU ($w{=}5$)  & 44.62 & 46.91 & 46.87 & 53.10 \\
      & GPU ($w{=}20$) & 43.96 & 44.83 & 46.77 & 53.18 \\
      & GPU ($w{=}50$) & 44.11 & 47.52 & 45.92 & 54.19 \\
      & CPU ($w{=}0$)  & 49.84 & 55.01 & 65.03 & 91.88 \\
      & CPU ($w{=}5$)  & 49.95 & 56.01 & 64.09 & 86.42 \\
      & CPU ($w{=}20$) & 48.35 & 58.14 & 65.56 & 89.32 \\
      & CPU ($w{=}50$) & 49.58 & 54.65 & 63.25 & 87.56 \\
    \midrule
    \multirow{8}{*}{Qwen2.5-7B}
      & GPU ($w{=}0$)  & 55.81 & 56.92 & 59.84 & 65.69 \\
      & GPU ($w{=}5$)  & 58.46 & 56.66 & 60.16 & 67.37 \\
      & GPU ($w{=}20$) & 56.97 & 58.30 & 59.99 & 64.59 \\
      & GPU ($w{=}50$) & 56.76 & 58.55 & 59.16 & 66.92 \\
      & CPU ($w{=}0$)  & 61.11 & 66.69 & 75.52 & 100.49 \\
      & CPU ($w{=}5$)  & 60.93 & 64.36 & 77.89 & 103.42 \\
      & CPU ($w{=}20$) & 63.98 & 64.04 & 72.19 & 98.02 \\
      & CPU ($w{=}50$) & 62.34 & 64.25 & 71.52 & 98.52 \\
    \bottomrule
  \end{tabular}
\end{table}

\begin{table}[H]
  \centering
  \footnotesize
  \setlength{\tabcolsep}{4pt}
  \caption{Query-to-full-answer (ms), GPU-resident vs.\ CPU-resident KV store,
  per encoding window $w$ and retrieval budget $k$.}
  \label{tab:cpu-offload-qta}
  \begin{tabular}{@{}ll rrrr@{}}
    \toprule
    Model & KV store & $k{=}5$ & $10$ & $20$ & $50$ \\
    \midrule
    \multirow{8}{*}{Llama-3.1-8B}
      & GPU ($w{=}0$)  & 250.94 & 225.20 & 207.18 & 201.23 \\
      & GPU ($w{=}5$)  & 289.15 & 279.15 & 264.81 & 274.02 \\
      & GPU ($w{=}20$) & 287.49 & 271.22 & 270.20 & 272.37 \\
      & GPU ($w{=}50$) & 284.91 & 271.34 & 262.61 & 269.38 \\
      & CPU ($w{=}0$)  & 253.07 & 231.44 & 224.38 & 242.59 \\
      & CPU ($w{=}5$)  & 297.65 & 287.13 & 280.02 & 312.66 \\
      & CPU ($w{=}20$) & 293.81 & 279.47 & 282.12 & 310.54 \\
      & CPU ($w{=}50$) & 290.21 & 279.84 & 279.30 & 308.77 \\
    \midrule
    \multirow{8}{*}{Mistral-7B}
      & GPU ($w{=}0$)  & 768.27 & 1808.69 & 1925.31 & 3026.33 \\
      & GPU ($w{=}5$)  & 262.98 & 281.87 & 338.52 & 583.65 \\
      & GPU ($w{=}20$) & 264.76 & 284.33 & 356.91 & 773.12 \\
      & GPU ($w{=}50$) & 257.31 & 271.65 & 342.00 & 902.23 \\
      & CPU ($w{=}0$)  & 775.13 & 1818.40 & 1942.44 & 3068.52 \\
      & CPU ($w{=}5$)  & 269.08 & 291.27 & 356.10 & 618.32 \\
      & CPU ($w{=}20$) & 269.12 & 297.82 & 375.54 & 809.00 \\
      & CPU ($w{=}50$) & 262.84 & 278.69 & 359.46 & 934.63 \\
    \midrule
    \multirow{8}{*}{Qwen2.5-7B}
      & GPU ($w{=}0$)  & 379.22 & 389.23 & 429.28 & 487.36 \\
      & GPU ($w{=}5$)  & 399.42 & 421.82 & 435.02 & 482.43 \\
      & GPU ($w{=}20$) & 392.52 & 405.92 & 420.61 & 444.02 \\
      & GPU ($w{=}50$) & 395.48 & 405.04 & 421.17 & 443.45 \\
      & CPU ($w{=}0$)  & 384.17 & 399.83 & 444.48 & 524.22 \\
      & CPU ($w{=}5$)  & 403.23 & 430.82 & 453.00 & 519.98 \\
      & CPU ($w{=}20$) & 398.48 & 411.86 & 433.56 & 479.74 \\
      & CPU ($w{=}50$) & 401.48 & 411.59 & 434.59 & 476.97 \\
    \bottomrule
  \end{tabular}
\end{table}

\section{Larger-Model (Qwen3-14B) Results}
\label{app:qwen3-14b}

\Cref{tab:qwen3-14b} reports a single headline operating point for Qwen3-14B;
here we give the full results. \Cref{tab:qwen3-14b-acc} sweeps overall accuracy
over the encoding window $w$ and retrieval budget $k$, \Cref{tab:qwen3-14b-percat}
breaks accuracy down by category, and
\Cref{tab:qwen3-14b-lat,tab:qwen3-14b-qta} give the full latency sweep for
GPU-resident and CPU-offloaded \sysname against Mem0. The trends match the
7--8B models: the $w{=}0$ configuration degrades sharply with $k$, a window
recovers it, Mem0 remains the accuracy ceiling, and \sysname holds a flat,
much lower latency.

\begin{table}[H]
  \centering
  \small
  \caption{Judged \locomo accuracy (\%) for Qwen3-14B vs.\ retrieval budget
  $k$: \sysname (per encoding window $w$) and Mem0 (prompt injection).}
  \label{tab:qwen3-14b-acc}
  \begin{tabular}{@{}l rrrr@{}}
    \toprule
    Method & $k{=}5$ & $10$ & $20$ & $50$ \\
    \midrule
    \sysname ($w{=}0$)  & 63.05 & 55.52 & 45.84 & 24.74 \\
    \sysname ($w{=}5$)  & 68.12 & 68.70 & 67.66 & 64.42 \\
    \sysname ($w{=}20$) & 68.25 & 69.61 & 69.87 & 68.90 \\
    \sysname ($w{=}50$) & 68.18 & 69.29 & 69.87 & 69.42 \\
    Mem0                & 74.94 & 77.79 & 78.44 & 79.09 \\
    \bottomrule
  \end{tabular}
\end{table}

\begin{table}[H]
  \centering
  \small
  \caption{Per-category judged \locomo accuracy (\%) for Qwen3-14B: \sysname
  (per encoding window $w$) vs.\ Mem0 (prompt injection).}
  \label{tab:qwen3-14b-percat}
  \begin{tabular}{@{}ll rrrr@{}}
    \toprule
    Category & Method & $k{=}5$ & $10$ & $20$ & $50$ \\
    \midrule
    \multirow{5}{*}{Single-hop}
      & \sysname ($w{=}0$)  & 70.27 & 64.45 & 54.82 & 29.73 \\
      & \sysname ($w{=}5$)  & 73.13 & 74.44 & 75.27 & 73.48 \\
      & \sysname ($w{=}20$) & 72.77 & 74.32 & 75.51 & 74.91 \\
      & \sysname ($w{=}50$) & 73.37 & 73.96 & 75.51 & 76.69 \\
      & Mem0                & 79.79 & 81.93 & 83.35 & 83.83 \\
    \midrule
    \multirow{5}{*}{Multi-hop}
      & \sysname ($w{=}0$)  & 50.71 & 44.68 & 39.36 & 22.70 \\
      & \sysname ($w{=}5$)  & 53.55 & 56.03 & 54.61 & 54.61 \\
      & \sysname ($w{=}20$) & 55.32 & 58.87 & 59.57 & 63.12 \\
      & \sysname ($w{=}50$) & 54.26 & 57.45 & 61.70 & 64.18 \\
      & Mem0                & 64.54 & 68.09 & 69.50 & 73.05 \\
    \midrule
    \multirow{5}{*}{Open-ended}
      & \sysname ($w{=}0$)  & 67.71 & 52.08 & 44.79 & 22.92 \\
      & \sysname ($w{=}5$)  & 66.67 & 68.75 & 67.71 & 61.46 \\
      & \sysname ($w{=}20$) & 67.71 & 66.67 & 68.75 & 62.50 \\
      & \sysname ($w{=}50$) & 65.62 & 65.62 & 62.50 & 59.38 \\
      & Mem0                & 65.62 & 69.79 & 71.88 & 62.50 \\
    \midrule
    \multirow{5}{*}{Temporal}
      & \sysname ($w{=}0$)  & 53.58 & 42.68 & 28.35 & 14.02 \\
      & \sysname ($w{=}5$)  & 68.22 & 64.80 & 59.19 & 50.16 \\
      & \sysname ($w{=}20$) & 67.91 & 67.60 & 64.49 & 60.12 \\
      & \sysname ($w{=}50$) & 67.60 & 68.54 & 64.49 & 57.94 \\
      & Mem0                & 74.14 & 77.88 & 75.39 & 76.95 \\
    \bottomrule
  \end{tabular}
\end{table}

\begin{table*}[t]
  \centering
  \small
  \caption{Qwen3-14B latency (ms): engine TTFT and end-to-end
  query-to-first-token, GPU-resident vs.\ CPU-offloaded \sysname and Mem0
  (Qdrant).}
  \label{tab:qwen3-14b-lat}
  \begin{tabular}{@{}l rrrr rrrr@{}}
    \toprule
    & \multicolumn{4}{c}{TTFT (ms)} & \multicolumn{4}{c}{Query-to-first-token (ms)} \\
    \cmidrule(lr){2-5}\cmidrule(lr){6-9}
    Method & $k{=}5$ & $10$ & $20$ & $50$ & $k{=}5$ & $10$ & $20$ & $50$ \\
    \midrule
    \sysname GPU ($w{=}0$)  & 26.41 & 26.43 & 26.99 & 28.19 & 78.74 & 78.16 & 81.25 & 82.09 \\
    \sysname GPU ($w{=}5$)  & 26.22 & 26.42 & 27.01 & 28.06 & 76.71 & 78.59 & 80.29 & 87.72 \\
    \sysname GPU ($w{=}20$) & 26.61 & 26.40 & 27.04 & 28.00 & 78.08 & 77.42 & 80.85 & 86.95 \\
    \sysname GPU ($w{=}50$) & 26.39 & 26.55 & 26.87 & 28.06 & 77.11 & 78.48 & 79.88 & 88.02 \\
    \sysname CPU ($w{=}0$)  & 26.54 & 27.78 & 28.28 & 31.53 & 80.72 & 88.21 & 97.73 & 131.55 \\
    \sysname CPU ($w{=}5$)  & 26.75 & 27.61 & 28.31 & 31.33 & 82.34 & 87.53 & 97.17 & 131.34 \\
    \sysname CPU ($w{=}20$) & 26.97 & 27.44 & 28.21 & 31.26 & 82.43 & 87.45 & 97.06 & 131.24 \\
    \sysname CPU ($w{=}50$) & 26.89 & 27.49 & 28.18 & 31.11 & 84.06 & 87.05 & 97.63 & 129.55 \\
    Mem0 (Qdrant)           & 59.16 & 66.57 & 84.47 & 140.52 & 153.33 & 163.47 & 186.12 & 257.76 \\
    \bottomrule
  \end{tabular}
\end{table*}

\begin{table}[H]
  \centering
  \small
  \caption{Qwen3-14B query-to-full-answer latency (ms), GPU-resident vs.\
  CPU-offloaded \sysname and Mem0. The $w{=}0$ runs are anomalous at large $k$.}
  \label{tab:qwen3-14b-qta}
  \begin{tabular}{@{}l rrrr@{}}
    \toprule
    Method & $k{=}5$ & $10$ & $20$ & $50$ \\
    \midrule
    \sysname GPU ($w{=}0$)  & 720.64 & 865.20 & 1379.48 & 7587.51 \\
    \sysname GPU ($w{=}5$)  & 544.92 & 604.22 & 664.73 & 796.33 \\
    \sysname GPU ($w{=}20$) & 546.17 & 600.72 & 665.69 & 778.32 \\
    \sysname GPU ($w{=}50$) & 533.79 & 595.76 & 638.68 & 750.38 \\
    \sysname CPU ($w{=}0$)  & 722.46 & 875.15 & 1396.10 & 7639.13 \\
    \sysname CPU ($w{=}5$)  & 550.67 & 613.01 & 681.43 & 838.60 \\
    \sysname CPU ($w{=}20$) & 550.63 & 610.73 & 682.32 & 822.25 \\
    \sysname CPU ($w{=}50$) & 540.47 & 603.95 & 655.75 & 791.44 \\
    Mem0 (Qdrant)           & 596.70 & 658.14 & 718.65 & 830.76 \\
    \bottomrule
  \end{tabular}
\end{table}

% ============================================================
\section{Document-RAG Benchmarks: QASPER and MultiHop-RAG}
\label{app:rag-benchmarks}

\begin{table*}[t]
  \centering
  \caption{Document-RAG benchmarks: \sysname (KV injection) vs.\ prompt
  injection on QASPER~\cite{qasper} and MultiHop-RAG~\cite{multihoprag}. Both
  arms use identical GPU-native Jasper retrieval, chunk size $c{=}1024$,
  encoding window $w{=}5$, and retrieval depth $k{=}15$, so the comparison
  isolates the injection mechanism rather than the retrieval backend. Judged
  accuracy uses the LLM judge of \Cref{sec:eval:setup}; EM is exact match and
  F1 is token-overlap F1 against the reference answers. TTFT is the vLLM engine
  time to first token, Q2FT is end-to-end query-to-first-token, and Q2A is
  query-to-full-answer; latencies are means in milliseconds.}
  \label{tab:rag-benchmarks}
  \begingroup
  \setlength{\tabcolsep}{5pt}
  \footnotesize
  \begin{tabular*}{\textwidth}{@{\extracolsep{\fill}}ll rrr rrr@{}}
    \toprule
    & & \multicolumn{3}{c}{Quality (\%)} & \multicolumn{3}{c}{Latency (ms)} \\
    \cmidrule(lr){3-5}\cmidrule(lr){6-8}
    Model & Method & Judged acc. & EM & F1 & TTFT & Q2FT & Q2A \\
    \midrule
    \multicolumn{8}{@{}l}{\emph{QASPER} (1{,}451 questions)} \\
    \addlinespace[1pt]
    Llama-3.1-8B & Prompt injection & 31.77 & 7.17 & 17.52 & 898.85 & 909.51 & 1359.09 \\
    Llama-3.1-8B & \sysname & 27.15 & 11.92 & 21.81 & \textbf{37.38} & \textbf{687.17} & \textbf{948.82} \\
    \addlinespace
    Mistral-7B & Prompt injection & 30.74 & 2.41 & 15.79 & 914.07 & 926.20 & 1277.29 \\
    Mistral-7B & \sysname & 24.67 & 1.45 & 14.05 & \textbf{36.30} & \textbf{604.76} & \textbf{991.79} \\
    \addlinespace
    Qwen2.5-7B & Prompt injection & 29.70 & 18.47 & 25.14 & 808.41 & 820.89 & 952.96 \\
    Qwen2.5-7B & \sysname & 27.08 & 16.06 & 23.34 & \textbf{35.83} & \textbf{261.15} & \textbf{423.31} \\
    \midrule
    \multicolumn{8}{@{}l}{\emph{MultiHop-RAG} (2{,}556 questions)} \\
    \addlinespace[1pt]
    Llama-3.1-8B & Prompt injection & 54.11 & 42.80 & 45.89 & 817.76 & 830.44 & 893.25 \\
    Llama-3.1-8B & \sysname & \textbf{54.97} & \textbf{44.29} & 45.20 & \textbf{36.58} & \textbf{425.19} & \textbf{525.82} \\
    \addlinespace
    Mistral-7B & Prompt injection & 55.16 & 31.96 & 34.29 & 827.47 & 839.73 & 995.40 \\
    Mistral-7B & \sysname & \textbf{55.79} & 25.08 & 28.56 & \textbf{35.34} & \textbf{411.38} & \textbf{676.01} \\
    \addlinespace
    Qwen2.5-7B & Prompt injection & 49.22 & 44.21 & 44.94 & 749.27 & 761.24 & 804.95 \\
    Qwen2.5-7B & \sysname & 42.61 & 39.44 & 39.98 & \textbf{34.47} & \textbf{220.42} & \textbf{261.36} \\
    \bottomrule
  \end{tabular*}
  \endgroup
\end{table*}

\locomo exercises \sysname on conversational memory, where each retrieved unit
is a short extracted fact. To test whether the mechanism generalizes to a
different shape of retrieved context, we additionally evaluate on two
document-RAG benchmarks: \textbf{QASPER}~\cite{qasper}, information-seeking
question answering over scientific papers (1{,}451 questions), and
\textbf{MultiHop-RAG}~\cite{multihoprag}, multi-hop question answering over a
news corpus (2{,}556 questions). \Cref{tab:rag-benchmarks} reports the results.

\para{Setup} The retrieved unit here is a fixed-size document chunk rather than
an extracted fact. Documents are split into $c{=}1024$-token chunks, each chunk
is encoded offline behind a window of $w{=}5$ preceding chunks, and every query
retrieves $k{=}15$ chunks, so a request conditions on ${\sim}15$K tokens of
retrieved context, an order of magnitude more than a \locomo request at
$k{=}50$. Both arms use the same GPU-native Jasper index and the same retrieved
chunks; they differ only in whether those chunks reach the model as prompt text
(prompt injection) or as injected KV (\sysname). This isolates KV injection
from the GPU-native-retrieval contribution, so these numbers are a
\emph{conservative} stand-in for a comparison against Mem0's default
configuration, which would additionally pay the CPU vector-store overhead of
\Cref{tab:vector-backend}.

\para{Latency} The latency advantage is far larger than on \locomo. Engine TTFT
falls from $749$--$914$\,ms under prompt injection to $34.5$--$37.4$\,ms under
\sysname, a $21.7$--$25.2\times$ reduction on every model and dataset, versus
$3.5$--$4.6\times$ on \locomo. This is exactly what the mechanism predicts:
prompt injection must prefill the retrieved context on every request, so its
cost scales with the ${\sim}15$K retrieved tokens, whereas KV injection
prefills only the query and its TTFT is essentially unchanged from the \locomo
setting ($16$--$18$\,ms there, $34$--$37$\,ms here, the residual difference
being the larger injected KV). End-to-end, \sysname lowers query-to-first-token
by $1.3$--$3.5\times$ and query-to-full-answer by $1.3$--$3.1\times$. The
end-to-end gains are smaller than the engine gains because at this chunk size
the retrieval and KV-composition path, not the engine, dominates \sysname's
critical path: on QASPER with Llama-3.1-8B, TTFT is $37$\,ms but
query-to-first-token is $687$\,ms. Reducing that composition cost is the
natural optimization target for large-chunk workloads.

\para{Quality} On MultiHop-RAG, \sysname \emph{matches or exceeds} prompt
injection on judged accuracy for Llama-3.1-8B ($54.97$ vs.\ $54.11$) and
Mistral-7B ($55.79$ vs.\ $55.16$), and leads on exact match for Llama-3.1-8B
($44.29$ vs.\ $42.80$). On QASPER, judged accuracy trails by $2.6$--$6.1$
points, but \sysname leads on both EM and F1 for Llama-3.1-8B ($11.92$/$21.81$
vs.\ $7.17$/$17.52$). This is consistent with the account in
\Cref{sec:eval:accuracy}: the accuracy cost of independent encoding comes from
missing cross-chunk attention, and a $1024$-token chunk already carries
substantial context of its own, so the deficit that shows up on \locomo's short
independently encoded facts largely disappears at document-chunk granularity.
Context-window encoding is doing the same work here that it does on \locomo,
just at a coarser unit.

Two results run against this trend and we report them as measured. Qwen2.5-7B
loses $6.6$ points of judged accuracy on MultiHop-RAG ($42.61$ vs.\ $49.22$)
and $4.8$ points of EM, the largest quality regression we observe anywhere, and
Mistral-7B loses $6.9$ points of EM on the same benchmark despite winning on
judged accuracy. We did not sweep $k$ or $w$ on these benchmarks, so we cannot
rule out that a larger encoding window closes these gaps as it does on \locomo;
the results here are a single operating point ($c{=}1024$, $w{=}5$, $k{=}15$)
chosen to match standard RAG practice rather than tuned per model.

% ============================================================
\section{Empirical equivalence: per-case results}
\label{app:equivalence}

\Cref{tab:equivalence} gives the per-user, per-query-type breakdown for the
empirical-equivalence test in \Cref{sec:eval:equivalence}.

\begin{table}[h]
\centering
\caption{Output comparison: KV injection vs.\ prompt injection with greedy
decoding (temperature~$=0$). Both runs use the IDENTICAL prompt token
sequence (encoded memory prefix + chat-templated user turn); the only
difference is whether the prefix's KV cache is supplied by
\texttt{MemoryKVConnector} or recomputed by standard prefill.
``Identical'' means token-for-token match;
``Sem.~equiv.'' means all facts are preserved with different phrasing.
The phrasing differences are consistent with floating-point
non-determinism in GPU reductions and do not indicate a violation of
Theorem~\ref{thm:equivalence}.}
\label{tab:equivalence}
\small
\resizebox{\columnwidth}{!}{%
\begin{tabular}{@{}lllcc@{}}
\toprule
\textbf{Model} & \textbf{User} & \textbf{Query Type} & \textbf{KV vs.\ PI} & \textbf{Factual Acc.} \\
\midrule
  \multirow{6}{*}{Llama-3.1-8B} & $u_1$ (143 tok) & Factual recall & Identical & 100\% \\
   &  & Open-ended & Sem.\ equiv. & 100\% \\
   & $u_2$ (111 tok) & Factual recall & Identical & 100\% \\
   &  & Open-ended & Sem.\ equiv. & 100\% \\
   & $u_3$ (111 tok) & Factual recall & Identical & 100\% \\
   &  & Open-ended & Sem.\ equiv. & 100\% \\
\midrule
  \multirow{6}{*}{Mistral-7B} & $u_1$ (131 tok) & Factual recall & Sem.\ equiv. & 100\% \\
   &  & Open-ended & Sem.\ equiv. & 100\% \\
   & $u_2$ (93 tok) & Factual recall & Identical & 100\% \\
   &  & Open-ended & Sem.\ equiv. & 100\% \\
   & $u_3$ (94 tok) & Factual recall & Sem.\ equiv. & 100\% \\
   &  & Open-ended & Identical & 100\% \\
\midrule
  \multirow{6}{*}{Qwen2.5-7B} & $u_1$ (121 tok) & Factual recall & Identical & 100\% \\
   &  & Open-ended & \textbf{Different} & 100\% \\
   & $u_2$ (89 tok) & Factual recall & Identical & 100\% \\
   &  & Open-ended & Sem.\ equiv. & 100\% \\
   & $u_3$ (89 tok) & Factual recall & Identical & 100\% \\
   &  & Open-ended & Identical & 100\% \\
\bottomrule
\end{tabular}%
}
\end{table}

% ------------------------------------------------------------
\section{Position independence and rotation fidelity}
\label{sec:eval:position}

\begin{table}[t]
\centering
\small
\caption{TTFT under chunk-order permutation. Memory consists of 64-token
chunks; the \emph{shuffled} configurations permute the chunks with a
deterministic per-trial seed. Prefix caching (row~2) yields a flat
near-decode latency when the chunk order repeats but collapses to
full-prefill cost under permutation (row~3). Chunked-RoPE composition
(row~4) achieves cache-hit latency \emph{without requiring the cache to
hit}. All values are mean TTFT in milliseconds over 3 trials; standard
deviations are below 5\,ms in every cell. \texttt{---} indicates
configurations not run for that model.}
\label{tab:shuffle}
\resizebox{\columnwidth}{!}{%
\begin{tabular}{llrrrr}
\toprule
Model & Configuration & 1K & 4K & 16K & 32K \\
\midrule
\multirow{4}{*}{Llama-3.1-8B-Instruct}
  & PI / no prefix cache / fixed     &  613 &  766 & 1652 & 3305 \\
  & PI / prefix cache / fixed        &  572 &  589 &  652 &  736 \\
  & PI / prefix cache / shuffled     &  614 &  767 & 1656 & 3311 \\
  & \textbf{Chunked-RoPE / shuffled} & \textbf{576} & \textbf{592} & \textbf{656} & \textbf{746} \\
\midrule
\multirow{4}{*}{Mistral-7B-Instruct}
  & PI / no prefix cache / fixed     &  586 &  737 & 1624 & --- \\
  & PI / prefix cache / fixed        &  544 &  561 &  623 & --- \\
  & PI / prefix cache / shuffled     &  583 &  739 & 1629 & --- \\
  & \textbf{Chunked-RoPE / shuffled} & \textbf{547} & \textbf{563} & \textbf{628} & --- \\
\midrule
\multirow{4}{*}{Qwen2.5-7B-Instruct}
  & PI / no prefix cache / fixed     &  593 &  726 & 1485 & --- \\
  & PI / prefix cache / fixed        &  556 &  564 &  597 & --- \\
  & PI / prefix cache / shuffled     &  593 &  727 & 1489 & --- \\
  & \textbf{Chunked-RoPE / shuffled} & \textbf{559} & \textbf{567} & \textbf{600} & --- \\
\bottomrule
\end{tabular}%
}
\end{table}

\Cref{thm:chunked-rope} states that a chunk stored pre-RoPE can be injected at
any virtual position with exactly the attention it would receive under prompt
injection. Two consequences are testable: the on-the-fly rotation must be
numerically lossless, and serving latency must depend on \emph{which} chunks are
active, not on the order in which they are composed. The rotation is lossless by
construction, and \Cref{sec:chunked-rope} verifies that re-rotating captured
pre-RoPE keys reproduces the model's post-RoPE keys to within bf16 noise; here
we test the operational consequence, order-independence.

\para{Setup} We tokenize a fixed corpus into 64-token chunks and, for memory
sizes $N\in\{1\text{K},4\text{K},16\text{K},32\text{K}\}$, use the first $N/64$
chunks. We measure engine TTFT under four configurations: prompt injection with
prefix caching \emph{disabled} (the unoptimized baseline); prompt injection with
prefix caching \emph{enabled} and a \emph{fixed} chunk order (the best case for
byte-prefix caching, which hits on every trial); the same but with chunks
\emph{permuted} per trial under a deterministic seed (defeating byte-prefix
reuse); and \sysname, which composes those \emph{same} permuted pre-RoPE chunks
with chunked RoPE and injects them. Each cell is the mean of 3 trials after 2
warm-ups, run in isolated subprocesses.

\para{Results} The pattern in \Cref{tab:shuffle} is uniform across models.
Prefix caching yields a flat, decode-dominated TTFT only while the chunk order
repeats; permuting the chunks collapses the cache and reverts prompt injection
to full-prefill cost that grows with $N$---on Llama at 32K, $3311$\,ms, $4.5\times$
the cache-hit floor and indistinguishable from running with no cache. \sysname,
composing the same permuted chunks on the fly, instead matches the cache-hit
floor to within about $1\%$ at every size and model: $4.44\times$ faster than
shuffled prompt injection at 32K on Llama, and $2.59\times$/$2.48\times$ at 16K
on Mistral/Qwen. On-the-fly rotation thus reaches the latency floor that
classical caching attains only when the prefix happens to repeat, confirming
that \sysname's cost is set by the active chunk set, not its arrangement. (The
one-time composition step is cached across queries that reuse the same memory;
we fold it into the end-to-end latency reported in \Cref{sec:eval:latency}.)

\end{document}